\documentclass[12pt]{amsart}

\usepackage[margin=1in]{geometry}
\usepackage[T1]{fontenc}
\usepackage[utf8]{inputenc}
\usepackage{lmodern}
\usepackage{microtype}
\usepackage{bm}
\usepackage{amsmath,amssymb,amsthm,mathtools}
\usepackage{enumitem}
\usepackage{xcolor}
\usepackage[colorlinks=true,linkcolor=blue,citecolor=blue,urlcolor=blue]{hyperref}
\usepackage[capitalize,noabbrev]{cleveref}
\usepackage[
 backend=biber,
 style=alphabetic,
 sorting=nyt,
 giveninits=true,
 maxnames=8,
 uniquelist=false
]{biblatex}
\usepackage{comment}
\usepackage{todonotes}
\usepackage{etoolbox}
\AtBeginEnvironment{quote}{\itshape\small} 
 
\AtBeginBibliography{\small}

\hypersetup{
 pdftitle={Uniformity in Finite Convolution Models},
 pdfauthor={Nilava Metya and Satyaki Mukherjee},
 pdfkeywords={finite convolution, loaded dice, entropy, extremal probability, Wasserstein distance}
}

\newtheorem{lemma}{Lemma}[section]
\newtheorem{theorem}[lemma]{Theorem}
\newtheorem{proposition}[lemma]{Proposition}

\newtheorem{conjecture}[lemma]{Conjecture}
\theoremstyle{definition}

\theoremstyle{remark}
\newtheorem{remark}[lemma]{Remark}

\makeatletter
\renewcommand\paragraph{\@startsection{paragraph}{4}{\z@}%
  \z@{-.5em}%
  {\normalfont\bfseries}}
\makeatother

\newcommand{\R}{\mathbb{R}}

\newcommand{\cM}{\mathcal{M}}
\newcommand{\DeltaS}[1]{\Delta_{#1}}
\newcommand{\norm}[2]{\left\lVert #1\right\rVert_{#2}}
\newcommand{\abs}[1]{\left|#1\right|}
\newcommand{\set}[1]{\left\{#1\right\}}
\newcommand{\E}[1]{\underset{#1}{\mathbb E}}
\newcommand{\KL}{D_{\mathrm{KL}}}
\newcommand{\supp}{\operatorname{supp}}

\newcommand{\sett}{\coloneqq}
\newcommand{\st}{%
  \nonscript\;
  \ifnum\currentgrouptype=16
    \;\middle|\;
  \else
    \;|\;
  \fi
  \nonscript\;}

\title{Approximate Uniformity in Finite Convolution Models}
\author{Nilava Metya}\thanks{authors listed in alphabetical order of surname}
\address{Department of Mathematics, Rutgers University--New Brunswick}
\email{nilava.metya@rutgers.edu}

\author{Satyaki Mukherjee}
\address{Department of Mathematics, National University of Singapore}
\email{satyaki.mukhyo@gmail.com}
\date{\today}
\begin{document}
\maketitle

\begin{abstract}
The problem of the uniform law being a sum of two independent distributions has been well studied. Here, we study the approximation of the uniform law to a sum of two distributions with fixed support, under the following discrepancies: the Manhattan distance, the Euclidean distance, the forward Kullback--Leibler divergence and the Wasserstein-one distance based on the line metric. The problem of membership of the uniform law in this model has been well studied. Explicit results are obtained in two models, one where one of the distributions is a Bernouli and the other when both distributions have the same support, including one conjecture. Finally, we show an application of our reconstruction method to a recent conjecture about the coding capacity of an additive noise channel.
\end{abstract}


\tableofcontents

\section{Introduction}
In \cite{kelly1950}, Kelly asked whether two biased six-sided dice can have a uniformly distributed sum. In technical terms, this problem can be rephrased as a factorization of the generating polynomials of the respective probability vectors. For a probability vector $\bm p\in \Delta_{n-1} \sett \set{(\mu_0,\cdots,\mu_{n-1}) \in \R_{\ge 0}^n \st \sum_i\mu_i=1}$ supported on $n$ points, define its polynomial generating function as $$G_{\bm p}(x) \sett \E{X\sim \bm p}[x^X] = \sum_i p_ix^i.$$ If $\bm p\in \Delta_{n-1}, \bm q\in \Delta_{m-1}$ and $X\sim \bm p,Y\sim \bm q$ are independent, then $X+Y\sim \bm p*\bm q$ where $*$ denotes the convolution. Then independence gives $G_{\bm p*\bm q} = G_{\bm p}G_{\bm q}$. This factorization of polynomials is what gives the answer to the originally question of Kelly negatively. The following paragraph gives the simple argument.

Suppose $\bm u = \bm p * \bm q \in \Delta_{10}$ for some $\bm p,\bm q\in \Delta_5$ where $\bm u = (1,\cdots,1)/11$ is the uniform distribution. Then $11G_{\bm r}(x) = (x^{11}-1)/(x-1)$. The roots of $G_{\bm r}$ are the ten non-real $11^{th}$ roots of unity. In particular, they are all non-real. If the factorization $G_{\bm r} = G_{\bm p}G_{\bm q}$ were true, then both $G_{\bm p}, G_{\bm q}$ would have only non-real roots. But this is not the case because both are polynomials of degree five with real coefficients whence they have at least one real root.

The central object of study in this paper is the compact convolution model $$\cM_{m,n} = \set{\bm p*\bm q \in \Delta_{m+n-2}\st \bm p\in \Delta_{m-1}, \bm q\in \Delta_{n-1}}.$$
$\cM_{m,n}$ is the set of laws of of the sum of two independent random variables, one supported on $\set{0,\cdots,m-1}$ and the other on $\set{0,\cdots,n-1}$. 

The membership in $\cM_{m,n}$ of the uniform law in $\Delta_{m+n-2}$ is studied in \cite{morrison2018}, and the complexity aspects of deconvolution of a general law are studied in \cite{bauschCubitt2016}. In the case of non-membership of the uniform law in $\cM_{m,n}$, the most natural question is to approximate the uniform law under a discrepancy $d$, namely, $\min_{\bm r\in \cM_{m,n}} d(\bm r, \bm u)$ where $\bm u\in \Delta_{m+n-2}$ is the uniform law. In this paper, $d\in \set{\ell_1, \ell_2, W_1^{\text{line}}, \KL}$. While $\KL$ is not a metric, it is a measure of discrepancy of general interest and falls under the general class of $f-$divergences. Please refer to Appendix~\ref{sec:w-dist-tutorial} for a short relevant tutorial on Wasserstein distances.

Specializing to $m=2$ we get the model $\cM_{2,n}$ which contains the laws of $B+X$ where $B\in\set{0,1}$ is a Bernoulli and $X\in\set{0,\cdots,n-1}$ is independent from $B$. Such a family exhibits an explicit description with one parameter. This interpretation of corrupting a main message $X$ with an independent Bernoulli $B$ is the exact setting used in a Bernoulli additive-noise sticky channel as defined by \textcite[Definition 1.1, Section 1.3.1]{BCPR2026}. Our one-parameter description of the model lets us settle a conjecture \cite[Conjecture 3.6]{BCPR2026} about the coding capacity of such a channel. Here is the one-parameter description in short. Suppose $\bm r = (r_0,\cdots,r_n)$ is a convolution with $\bm r = (a,1-a)*\bm q$ where $a>0$. Then we have that $t=(1-a)/a$ is a root of the polynomial $\E{Z\sim \bm r}\left[(-t)^{n-Z}\right] = (-t)^nG_{\bm r}(-1/t)$. Conversely a nonnegative root of this polynomial lets us recover Bernoulli part $(a,1-a)$ and the main distribution $\bm q$ recursively via non-negativity constraints. As an immediate consequence, the uniform law on $\set{0,\cdots,n}$ belongs to $\cM_{2,n}$ exactly when $n$ is
odd, and in that case its ordered factorization is unique. 

For the problem of projection of the uniform law onto this non-convex model, we summarize our results in the following statement. The explicit statements are mentioned later in Section~\ref{sec:proj-m2n}.

\begin{theorem}
Let $\bm u = (n+1)^{-1}\bm 1\in \Delta_{n}$ and suppose $n$ is even. Let $\delta_n = \tau_n^n/(1+\tau_n)$ $\tau_n$ be the unique root in $(0,1)$ of the equation $$(n+1)\tau^n + (n-1) \tau^{n+1} =2.$$
Then optimal values for the projection problems are \begin{align*}
\min_{\bm r\in \cM_{2,n}} \norm{\bm r-\bm u}{2}^2 = \frac{1}{n(n+1)(n+2)}, &\qquad   \min_{\bm r\in \cM_{2,n}} \KL(\bm r\|\bm u) = \log \frac{n+1}{\sqrt{n(n+2)}},\\
\min_{\bm r\in \cM_{2,n}} \norm{\bm r-\bm u}{1} = \delta_n,& \qquad \min_{\bm r\in \cM_{2,n}} W_1^{\text{line}}(\bm r,\bm u) = \delta_n/2.
\end{align*}
\end{theorem}

While this is easy to prove for $\ell_2$ projection, the $\KL$ result requires substantially more work. The minimization $\KL$ problem is equivalent to entropy maximization because we are comparing against the uniform law $\bm u$. A nondegenerate factorization supplies a negative root $-\rho$ of $G_{\bm r}$. At fixed $\rho$, Gibbs’ variational inequality (or first order optimality condition) reduces entropy maximization to the exponential family
$$r_i\propto \exp\left(\lambda(-\rho)^i\right), \qquad Z_\rho(\lambda) = \sum_{i=0}^n \exp\left(\lambda(-\rho)^i\right).$$
The challenging step is to show that $\min_\lambda Z_{\rho}(\lambda)$ is strictly increasing with respect to $\rho\in (0,1]$. We prove this through a weighted alternating exponential-sum inequality in Section~\ref{sec:technical-KL}.

A second structurally convenient regime is the diagonal model $\cM_{n,n}$. In this model, the $\ell_2$ and the $\KL$ projection theorems are readily available from recent literature. The $\ell_2$ projection formula was proven by \textcite[Theorem 1]{asgarliEtAl2024}, while the $\KL$ projection formula follows from \textcite[Theorem 2.1]{kovacevic2021}. For the $W_1^{\text{line}}$ metric, we prove $$\min_{\bm r\in \cM_{n,n}} W_1^{\text{line}}(\bm r,\bm u) = (n-1)\left(\sqrt{\frac{2n}{2n-1}}-1\right).$$
The proof is a standard technique used for tackling linear programs. For $\ell_1$ projection, we construct a family at distance $1/(2n-1)$ from uniform for every $n\ge 3$. The matching lower bound is left as a conjecture.

\subsection{An application to information theory}
A natural interpretation of the model $\cM_{2,n}$ is to look at random variables $X\in\set{0,\cdots,n-1}$ corrupted by an independent Bernouli noise $B\in \set{0,1}$. A lot of the structural analysis in this document, especially the one parameter description of Theorem~\ref{thm:M2n-membership}, utilizes the fact that the corruption is done by a one-bit error, namely $B$, and not as much on the support of the main message $X$. Thus, some of our methods also extend to the case when $X$ is supported on $\mathbb N\cup\set{0}$. The corresponding model would be then called $\cM_{2,\mathbb N\cup\set{0}}$.

The above is just a Bernoulli convolution, which arise naturally in communication channels with synchronization errors. In the Bernoulli additive-noise sticky channel, described by \textcite[Definition~1.1, Section~1.3.1]{BCPR2026}, each maximal run of identical binary symbols is independently extended by one symbol with probability $p$ and otherwise left unchanged. Encoding a run by its length $X$ gives the additive channel $Z=X+B$, where $X$ is a positive integer and $B \sim \operatorname{Bernoulli}(p)$ is independent of $X$. If input run lengths are restricted to $\set{1,\cdots,n}$ and $q_j=\mathbb P[X=j+1]$, then the law of $Z-1$ is $\bm r=(1-p,p)*\bm q$. Thus the possible output laws form the subfamily of $\cM_{2,n}$ obtained by fixing the Bernoulli factor. Moreover, $I(X;Y)=H(\bm r)-h(p)$, where $h(p)=H(B)$ is the Bernoulli entropy. Since $D_{\mathrm{KL}}(\bm r\Vert\bm u)=\log(n+1)-H(\bm r)$, the fixed-bias KL projection onto this subfamily determines the largest mutual information per run.

The coding capacity additionally accounts for the number of input symbols consumed by each run. It is therefore characterized by $C(p)=\sup_X I(X;X+B)/\E{}X$, where the supremum ranges over positive-integer-valued random variables $X$ with finite mean; see \cite[Lemma~1.2]{BCPR2026}. \cite[Theorem~3.1]{BCPR2026} obtained an exponential-family upper bound for this quantity and conjectured about its optimality for $1/2<p\le 1$ in \cite[Conjecture~3.6]{BCPR2026}. The remaining issue is deconvolution. A proposed run-length observation $Z$ must be shown to arise from an actual run-length $X$. We prove a criterion that establishes this property both for finite entropy optimizers and for the geometrically decaying distributions arising in the capacity problem.

Indeed Lemma~\ref{lem:positive-deconvolution} proves that a probability law $\bm r$ on $\mathbb N\cup\set{0}$ of the form $r_j\propto z^j\exp\set{-b(-\rho)^j}$ such that $\E{X\sim \bm r}(-\rho)^X$ can be deconvolved into an independent sum $X+B$ where $B$ is a Bernoulli with parameter $(1+\rho)^{-1}$ and $X$ realizes each non-negative integer with positive probability. Applied to the capacity problem, it proves the conjectured capacity formula and yields a unique optimizing input with strictly positive probabilities and an exponentially decaying tail. The deterministic endpoint $p=1$ is treated separately.

\begin{theorem}[{\cite[Conjecture~3.6]{BCPR2026}}]
\label{conj:BCPR-capacity}
Let $B\sim\operatorname{Bernoulli}(p)$, where $1/2<p< 1$, and 
$$C(p) \coloneqq \sup_{\substack{X\perp B\\
\mathbb P[X\in\mathbb N]=1\\
\E{} X<\infty}} \frac{I(X;X+B)}{\E{} X},$$
with mutual information measured in bits. Under the convention $0\log_2 0=0$, write $h(p)=H(B)=-p\log_2 p-(1-p)\log_2(1-p).$
For each $\beta\in\mathbb R$, define $\lambda(p,\beta)$ by the equation
$$\sum_{y=1}^{\infty} 2^{-\lambda(p,\beta)(y-p)-h(p) -\beta\left(-\frac{1-p}{p}\right)^{y-1}}=1.$$
Then
$$C(p)=\inf_{\beta\in\mathbb R}\lambda(p,\beta).$$
\end{theorem}
Their original conjecture includes $p=1$ in the range, unlike what is stated above. The conjecture for the range $1/2<p<1$ is proven in Theorem~\ref{dext:thm:capacity} and the endpoint $p=1$ case is discussed in Remark~\ref{dext:rem:capacity-endpoint}. The proof and formal statement of the above deconvolution argument is stated in Lemma~\ref{lem:positive-deconvolution} and uses the same construction argument as Theorem~\ref{thm:M2n-membership}, described above, but for an infinite-support case. The bound is proven using the same idea as Gibbs variational inequality in Proposition~\ref{prop:gibbs}.

\subsection{Related work}

The exact loaded-dice problem begins with \textcite{kelly1950}. \textcite{chenRaoShreve1997} showed that two dice of the same order cannot have a uniform sum. \textcite[Main Theorem and Main Corollary]{gasarchKruskal1999} later characterized fair tuples
through their ``nice'' dice and the uniqueness of representations of totals. Respectively \cite[Corollaries~7 and~8]{gasarchKruskal1999} give the even-order and repeated-order obstructions. \textcite[Theorem~5.1]{morrison2018} completed this structural picture with an explicit construction of all fair sacks. These results concern exact uniformity. Our projection problems ask instead for the closest attainable law when exact uniformity fails.

The factorization viewpoint also connects the problem with discrete tiling. For factor polynomials with coefficients in $\set{0,1}$, exact factorization amounts to a direct-sum tiling problem for intervals or cyclic groups. \textcite{covenMeyerowitz1999}. develop the associated polynomial and cyclotomic methods Since our factors may have arbitrary nonnegative real coefficients, the tiling analogy is informative but does not yield the weighted extremal results considered here. In a more directly probabilistic direction, a recent work of \textcite{klich2026} studies positive factorization and its local stability for aggregate counting laws; see \cite[Remark~2.6 and Theorem~4.5]{klich2026}.

Approximation to a fair total already appears in \textcite{gasarchKruskal1999}, who asked for the $\ell_2$-approximation to uniformity and reported numerical experiments in \cite[pp.~137--138]{gasarchKruskal1999}. \textcite[Theorem~1]{asgarliEtAl2024} subsequently obtained the exact optimizer and proved its uniqueness up to swapping. For the $\KL$-projection, \cite[Theorem~2.1, equation~(2.3)]{kovacevic2021} on maximum-entropy for sums on the same support implies the exact equal-support two-sum result. Our negative-root entropy inequality addresses the different, asymmetric $2$-by-$n$ setting.

Finally, prescribed-support membership is related to computational questions about divisibility and decomposability of finite distributions. \textcite{bauschCubitt2016} distinguish divisibility into identically distributed factors from decomposition into arbitrary nontrivial factors. The former is polynomial-time decidable, whereas the latter is NP-complete; see \cite[Definitions~52 and~65, Theorems~54 and~77, and Lemma~68]{bauschCubitt2016}.

\subsection{Acknowledgements.} We would like to thank ChatGPT $5.6$ Pro for suggesting an alternative reduced lemma whose consequence is the $\KL-$projection problem in $\cM_{2,n}$, although it gave an incorrect reference and failed to produce a correct proof of the reduced lemma independently. However it was an immense help because this lemma was indeed correct and was eventually proven by the authors. We have also used ChatGPT for editing and language. NM would also like to thank Budhaditya Halder, Bora \c{C}alim and Jonathan Scarlett for helpful discussions.

\section{Membership and recronstruction in $\cM_{2,n}$}

\subsection{Membership}
We begin with exact membership of the uniform law in $\cM_{2,n}$. For $\bm p\in\DeltaS{m-1}$ and $\bm q\in\DeltaS{n-1}$, convolution is given by
\begin{equation}
(\bm p*\bm q)_k
=\sum_{i=\max\set{0,k-n+1}}^{\min\set{m-1,k}}p_iq_{k-i},
\qquad 0\leq k\leq m+n-2.
\label{eq:convolution}
\end{equation}
The identity $G_{\bm p*\bm q}=G_{\bm p}G_{\bm q}$ translates membership in $\cM_{m,n}$ exactly into factorization of $G_{\bm r}$. Namely, that the two factors have degrees at most $m-1$ and $n-1$, nonnegative coefficients, and value $1$ at $x=1$. This also gives a useful existence fact for the approximation problems.

Nonnegativity has another elementary consequence. Since no cancellation is possible,
\begin{equation}
\supp(\bm p*\bm q)=\supp(\bm p)+\supp(\bm q),
\label{eq:support-sumset}
\end{equation}
where the right side is the sumset of subsets of the integers. Vanishing coordinates therefore detect combinatorial restrictions on the factors and later appear as natural boundary strata.

The uniform generating polynomial is a normalized geometric sum. Its roots are roots of unity, so parity gives immediate obstructions, while divisibility yields useful constructions We state two consequences that will help the later approximation results.

\begin{proposition}
\label{prop:classical-fairness}
Let $\bm u=(m+n-1)^{-1}\bm 1$ on
$\set{0,\ldots,m+n-2}$.
\begin{enumerate}[label=\textup{(\roman*)}]
\item If $m$ and $n$ are both even, then $\bm u\notin\cM_{m,n}$.
\item If $m\mid n-1$ or $n\mid m-1$, then
   $\bm u\in\cM_{m,n}$.
\end{enumerate}
\end{proposition}

\begin{proof}
For (i), the positive leading coefficient of $G_{\bm u}$ forces both factor polynomials to have their full prescribed degrees $m-1$ and $n-1$. If $m,n$ are even, both degrees are odd, so both real factor polynomials have a real zero. On the other hand,
$$
G_{\bm u}(x)=\frac{1+x+\cdots+x^{m+n-2}}{m+n-1}
$$
has as its zeros the nontrivial $(m+n-1)$st roots of unity. Since $m+n-1$ is odd, none is real, a contradiction. This is a special case of \cite[Corollary~7]{gasarchKruskal1999}.

For (ii), WLOG say $n-1=mL$. With $\Psi_d(x)=(x^d-1)/(x-1)$, the identity $\Psi_{m(L+1)}(x)=\Psi_m(x)\Psi_{L+1}(x^m)$ proves the result. This is a special case of \cite[Corollary~4.2]{morrison2018}.
\end{proof}

\begin{remark}
Divisibility is a convenient construction, not a characterization. For example,
$$
(1+x)(1+x^6)\,(1+x^2+x^4)=1+x+\cdots+x^{11}
$$
gives a fair $8$-by-$5$ pair after normalization, although neither $8\mid4$ nor $5\mid7$.
\end{remark}

The factorization fact above makes the membership situation in $\cM_{2,n}$ rigid and provides the actual obstructions. While the following result immediately follows from Proposition~\ref{prop:classical-fairness}, we still prove it to demonstrate to the reader the power of generating polynomials and its actual role in the rigidity of $\cM_{2,n}$.

\begin{proposition}[Exact fairness in $\cM_{2,n}$]
\label{prop:M2n-fairness}
Let $n\geq2$ and let $\bm u=(n+1)^{-1}\bm 1\in\DeltaS n$. Then $\bm u\in\cM_{2,n}$ if and only if $n$ is odd. When $n$ is odd, the ordered factorization is unique and is given by $\bm p =(1,1)/2, \bm q = (1,0,1,\cdots,1)\cdot 2/(n+1)$.
\end{proposition}

\begin{proof}
Every nondegenerate two-point law has a generating polynomial with one negative real root. The roots of $1+x+\cdots+x^n$ are the nontrivial $(n+1)$st roots of unity, and among them a negative real root occurs exactly when $n$ is odd -- that root is $-1$. Consequently any factorization of $\bm u$ forces $G_{\bm p}(x)=(1+x)/2$, after normalization at $1$. Polynomial division gives the abovementioned vector $\bm q$, whose entries are nonnegative and sum to one. Direct multiplication gives $G_{\bm p}(x) G_{\bm q}(x) = (1+x+\cdots+x^n)/(n+1)$. The root determines the first factor, and division determines the second, so the factorization is unique. A degenerate two-point factor is impossible because the uniform output has two positive endpoint coordinates.
\end{proof}

Thus odd $n$ is settled exactly. When $n$ is even, uniformity fails and the rest of the paper asks how close one can get to fairness.

\subsection{Roots and Reconstruction}

Recall that $\cM_{2,n}$ is the set of laws of sums of the form $B+X$ where $B\in\set{0,1}$ is Bernoulli, $X\in \set{0,\cdots,n-1}$ and both are independent. We shall show that once the Bernouli factor is chosen, the other factor can be recovered one coordinate at a time. Say the probability vectors for $B,X$ are respectively $(a,1-a), \bm p$. Let's encode the first factor by its odds ratio $t=(1-a)/a.$

For a vector $\bm r\in \Delta_n$ define the polynomials 
\begin{align}
Q_0(t) &= (1+t)r_0,\notag\\
Q_k(t) &= (1+t)r_k - tQ_{k-1}(t), \qquad 1\le k\le n-1.\label{eq:recover-q}
\end{align}
and the generating polynomial \begin{align}
\label{eq:gen-r}R_{\bm r} = r_n - tr_{n-1} + \cdots+(-t)^nr_0 = (-t)^nG_{\bm r}(-1/t).
\end{align}

The leftmost convolution equation gives $q_0$, each subsequent equation gives the next $q_k$, and the final coordinate closes the recursion with a single polynomial equation in $t$. The following theorem makes this reduction precise.

\begin{theorem}[Membership and factor reconstruction]
\label{thm:M2n-membership}
There is a bijection between ordered factorizations $\bm r= (a,1-a)*\bm q\in \Delta_n$ with $a>0$, and values $t\ge 0$ satisfying $R_{\bm r} = 0, Q_k(t)\ge 0~\forall~ 0\le k<n$. The correspondence is
$$a=\frac{1}{1+t}, \qquad q_k=Q_k(t).$$
\end{theorem}

\begin{proof}
Suppose first that $\bm r=(a,1-a)*\bm q$ with $a>0$, and set $t=(1-a)/a$. Writing the convolution equations gives
\begin{align*}
r_0&=aq_0,\\
r_k&=aq_k+(1-a)q_{k-1}\quad(1\leq k\leq n-1),\\
r_n&=(1-a)q_{n-1}.
\end{align*}
The first $n$ equations above give \eqref{eq:recover-q} and $q_k=Q_k(t)$. Iterating the recursion yields
\begin{equation}
Q_k(t)=(1+t)\sum_{j=0}^{k}(-t)^{k-j}r_j.
\label{eq:Q-closed}
\end{equation}
Using the equation $tQ_{n-1}(t) = r_n(1+t)$, substituting $k=n-1$ into the above and multiplying by $t$ gives precisely $R_{\bm r}(t)=0$. Therefore a factorization of $\bm r$ gives a value $t\ge 0$ which is a root of $R_{\bm r}$  and evaluates to a non-negative value on each $Q_k$.

Conversely, assume $t\ge 0$ is a root of $R_{\bm r}$ such that $Q_{k}(t)\ge 0~\forall~0\le k<n$. Define $a = (1+t)^{-1}$ and $\bm q = (Q_0(t),\cdots,Q_{n-1}(t))$. The recursion recovers all but the final convolution coordinate $r_n$. But the equation $R_{\bm r}(t)=0$ recovers the value of $r_n$. Convolution by $(a,1-a)$ preserves total mass, so $\sum_i r_i=1$ forces $\sum_iq_i=1$. Together with the assumed inequalities, this forces $\bm q\in \DeltaS{n-1}$. The above equations also show that a fixed $t$ admits no second choice of $\bm q$. 
\end{proof}

\begin{remark}
The remaining case for $a=0$ exists exactly when $r_0=0$, and then $\bm q=(r_1,\ldots,r_n)$. Equivalently, it is the reversal of the $t=0$ branch.
\end{remark}

We may therefore think of an ordered factorization as a negative root of the generating polynomial $G_{\bm r}$ of $\bm r$. But this is only an `admissible' root because polynomial division by a negative root need not have nonnegative coefficients which is needed for a valid probability law, and the inequalities $Q_k(t) \ge 0$ are exactly what rules this out.

Using the same technique as above, one can prove a deconvolution for the extreme case of $\bm r$ given by the Gibbs inequality, either when $n$ is even or the support is the set of all non-negative integers. Precisely speaking, let $k\in \mathbb N\cup \set{0,\infty}, b>0, \rho\in(0,1], z\in (0,1]$ and consider the density $\bm r$ supported on $\mathbb Z\cap [0,2k+1)$ $$r_j\propto z^j\exp\set{-b(-\rho)^j}.$$
If $k=\infty$, we will require $\rho<1,z<1$. Let $A$ be the normalizing constant that makes $r_j$'s sum to $1$.

\begin{lemma}
\label{lem:positive-deconvolution}
Suppose $k,\rho,b,z$ are as above such that $\E{Z\sim \bm r}[(-\rho)^Z] = 0$. Then 
\begin{enumerate}
\item there is a unique probability law $\bm q$ such that $\bm r = (a,1-a)*\bm q$ where $a=\rho/(1+\rho)$;
\item all coordinates of $\bm q$ are strictly positive;
\item if $k=\infty$, then $q_j \le Ae^b\rho z^{j+1}/[a(1-\rho z)]~\forall~j\ge 0$.
\end{enumerate}
\end{lemma}
\begin{proof}
If $k<\infty$ and $n=2k$, one can invoke Theorem~\ref{thm:M2n-membership} with $t=1/\rho$ after verifying the polynomial constraints. Let $a = \rho/(1+\rho)$. We verify these as follows. $R_{\bm r}(t) = (-t)^n\E{Z\sim \bm r}[(-t)^{-Z}] = (-\rho)^{-n}\E{Z\sim \bm r}[(-\rho)^{Z}] = 0$ by the assumption in the theorem.

\eqref{eq:Q-closed} gives the closed form $Q_\ell(t) = (1+1/\rho)\sum_{j=0}^\ell(-\rho)^{j-\ell}r_j$ and we want to verify that this quantity is positive. Therefore we precisely want to prove that the partial sums $$S_\ell \sett \sum_{j=0}^\ell(-1)^j(\rho z)^j \exp\set{-b(-\rho)^j}$$ alternate in signs because $a(-\rho)^\ell Q_\ell(t) = AS_\ell$. Write $w=\rho z$, $u_j = w^{2j}\exp\set{-b\rho^{2j}}$ and $v_j = w^{2j+1}\exp\set{b\rho^{2j+1}}$. In this notation $$S_\ell = \sum_{0\le i\le \ell, 2\mid i} u_{i/2} - \sum_{0\le i\le \ell, 2\nmid i} v_{(i-1)/2}.$$
Consider $$\frac{v_j}{u_{j+1}} = w^{-1}\exp\set{b\rho^{2j+1}(1+\rho)} > 1, \qquad \frac{v_j}{u_j} = w \exp\set{b\rho^{2j}(1+\rho)}.$$
Note that the second ratio is decreasing in $j$, and strictly decreasing when $\rho<1$. The assumption that $\E{Z\sim \bm r}(-\rho)^Z = 0$ implies that $$\sum_{j=0}^k u_j = \sum_{j=0}^{k-1}v_j.$$
Therefore $$ S_{2\ell} = u_0 + \sum_{j=0}^{\ell-1} (u_{j+1}-v_j)$$ is strictly decreasing in $\ell$ as $u_{j+1}<v_j$. This equals zero at $\ell=k$. Therefore, it must be strictly positive for $0\le j < k$. This shows that all the even-indexed partial sums $S_{2\ell}$ are positive whence $Q_{2\ell}(t) > 0~\forall 0\le \ell <k$.\\
Now focus on the odd sums $$S_{2\ell-1} = \sum_{j=0}^{\ell-1} (u_j-v_j) = \sum_{j=0}^{\ell-1} u_j\left(1-w\exp\set{b(1+\rho)\rho^{2j}}\right).$$
Since $\exp\set{b(1+\rho)\rho^{2j}}$ decreases with $j$, the signs of $u_j-v_j$ can change at most once from negative to positive. But $S_{2k-1} = (u_0-v_0) + \cdots (u_{k-1}-v_{k-1}) = -u_{k} <0$. This means that every prefix sum of the sequence $\set{u_j-v_j}_{k=0}^{k-1}$ should be negative. This prefix sum is precisely $S_{2\ell-1}$. This proves $S_{2\ell-1}<0$ whence $Q_{2\ell-1}(t) > 0~\forall~1\le \ell < k$. By Theorem~\ref{thm:M2n-membership}, there is a unique factorization $\bm r = (a,1-a)*\bm q$.

For the case of $k=\infty$, all series converge absolutely. By the same calculations as above, the even partial sums decrease strictly to zero and hence are positive. The terms $u_j-v_j$ that sum to the odd partial sums can change sign at most once from negative to positive since $v_j/u_j$ decreases strictly. Further since $u_j-v_j$ sum to zero (due to the expectation polynomial constraint), one of them must be negative. This completes the argument that the odd partial sums are all negative. The recursive structure argument of Theorem~\ref{thm:M2n-membership} still works here and gives a $\bm q$ such that $\bm r=(a,1-a)*\bm q$. Here is a proof for exponentially decaying sums:
$$q_j=-\frac{(-\rho)^{-j}}a \sum_{\ell=j+1}^{\infty}(-\rho)^\ell r_\ell \le \frac{Ae^b}{a}\rho^{-j} \sum_{\ell=j+1}^{\infty}(\rho z)^\ell =\frac{Ae^b\rho z^{j+1}}{a(1-\rho z)}.$$
\end{proof}

\section{Projections in the $2-$by$-n$ model}\label{sec:proj-m2n}
Throughout this section, $$\bm u = \frac1{n+1}\bm 1 \in \Delta_n, \qquad \bm p=(a,1-a) \in \Delta_1, \qquad \bm q \in \Delta_{n-1}, \qquad \bm r = \bm p*\bm q\in \Delta_n.$$
We have already discussed in Proposition~\ref{prop:M2n-fairness} that when $n$ is odd, $\bm u\in \cM_{2,n}$. So in the rest of the section, we discuss about projections of $\bm u$ onto $\cM_{2,n}$ under the discrepancies $\ell_2, \KL, \ell_1, W^{\text{line}}_1$. We naturally divide section this into two parts. One discusses the $\ell_2,\KL$ projections since they are smooth in nature. The other part discusses the $\ell_1, W_1^{\text{line}}$ metrics which are polyhedral.

\subsection{$\ell_2$ and forward $\KL$}
For the quadratic problem \begin{align}\label{eq:l2min}
\min_{\bm r\in \cM_{2,n}}\norm{\bm r-\bm u}{2}^2
\end{align} the following calculation shows that the above projection is equivalent to $\min_{\bm r\in \cM_{2,n}}\norm{\bm r}{2}^2$: $$\bm r\in \Delta_n\implies\norm{\bm r-\bm u}{2}^2 = \norm{\bm r}{2}^2 + \norm{\bm u}{2}^2 - 2\bm u^\top \bm r = \norm{\bm r}{2}^2 - \frac{1}{n+1}.$$
Lemma~\ref{lem:quad-reduction} shows that the optimal $\bm p=(a,1-a)$, for any given $\bm q\in \Delta_{n-1}$, occurs at $a=1/2$. This then makes the final problem an optimization over $\bm q\in \Delta_{n-1}$ which is simply an $\ell_2$ projection. Note that the new objective is to minimize the quadratic form $\frac12\bm q^\top A_n \bm q$ over the simplex $\Delta_{n-1}$ where $(A_n)_{ii}=1$, $(A_n)_{i,i+1} = (A_n)_{i+1,i} = \frac12$ and other entries are zero. This is a positive definite form because expanding out gives a positive expression. 

\begin{theorem}
For even $n$, the optimal value of \eqref{eq:l2min} is $[n(n+1)(n+2)]^{-1}$ with solutions \begin{equation}\label{eq:l2solution}
\bm p^*=\left(\frac12,\frac12\right), \qquad \bm q^* = \frac{2}{n(n+2)}\left(n,2,n-2,4,\cdots,2,n\right).
\end{equation}
\end{theorem}
\begin{proof}
Lemma~\ref{lem:quad-reduction} reduces the problem to minimizing $\norm{(1/2,1/2)*\bm q}{2}^2$ over $\bm q\in \Delta_{n-1}$. Note that $\norm{(1/2,1/2)*\bm q}{2}^2 = \frac12\bm q^\top A_n\bm q$ where $A_n$ is stated above. The first order optimality condition gives $A_n\bm q = \lambda \bm 1$ for some $\lambda\in\R$, via Lagrange multipliers. Solving this system gives the solution $\bm q^*$ mentioned above. This is the minimizer because if $\bm q\in \Delta_{n-1}$ then $\bm h\sett \bm q-\bm q^*$ satisfies $\bm 1^\top \bm h=0$ whence $\bm q^\top A_n\bm q-(\bm q^*)^\top A_n(\bm q^*) = \bm h^\top A_n\bm h$ which can be zero only when $\bm h=0$ since $A_n\succ 0$. The optimal distance is a direct calculation.
\end{proof}

Now let's move on to projection under $\KL$. Namely, we want to solve \begin{equation}\label{eq:klmin}
\min_{\bm r\in \cM_{2,n}} \KL(\bm r \| \bm u).
\end{equation}
Note that $\KL(\bm r\|\bm u) = \log (n+1) - H(\bm r)$ where $H$ is the Shannon entropy given by $H(\bm r) = -\sum r_i\log r_i$. Therefore this problem is equivalent to maximizing the entropy $H(\bm r)$ of $\bm r$ over $\cM_{2,n}$. We use the convention $0\log 0 = 0$.

The entropy problem surprisingly has the same optimizer as the quadratic problem, but for a different reason. Since we want to optimize over $\bm r\in \cM_{2,n}$, there is negative root $-\rho$ of $G_{\bm r}$ by Theorem~\ref{thm:M2n-membership}. Namely,
$$\sum_{i=0}^n r_i(-\rho)^i = 0.$$
For a fixed $\rho$, define $$Z_\rho(\lambda) \sett \sum_{i=0}^n \exp\set{\lambda(-\rho)^i}, \qquad \mathcal Z(\rho) \sett \min_{\lambda\in \R} Z_\rho(\lambda).$$
By Proposition~\ref{prop:gibbs}, $$H(\bm r) \le \log Z_\rho(\lambda) - \lambda \sum_{i=0}^nr_i(-\rho)^i = \log Z_\rho(\lambda) ~\forall~\lambda \in\R.$$
Therefore $H(\bm r) \le \log \mathcal Z(\rho)$. Therefore it is enough to bound $\mathcal Z(\rho)$ over $\rho > 0$. Note that replacing $\rho$ with $1/\rho$ leaves the entropy unchanged whence it is enough to restrict to $\rho\in(0,1]$. This leads to the following lemma.

\begin{lemma}\label{lem:negative-root-entropy}
Let $n\in 2\mathbb N, \bm r\in \Delta_n,\rho>0$ be such that $G_{\bm r}(-\rho) = \E{Z\sim \bm r}(-\rho)^Z = 0$. Then $H(\bm r) \le \log \sqrt{n(n+2)}$. Equality holds iff $\rho=1, \bm r = \left((n+2)^{-1}, n^{-1}  , (n+2)^{-1}, \cdots, n^{-1} , (n+2)^{-1}\right).$
\end{lemma}

Here is a proof sketch with a detailed proof in Section~\ref{sec:completion-entropy}. The central non-trivial fact is Lemma~\ref{lem:partition-monotonicity}, namely that $\mathcal Z(\rho)$ is strictly increasing on $(0,1]$. Modulo this fact, we have $$H(\bm r) \le \log \mathcal Z(1) = \log\left(\min_{\lambda\in \R} \set{\frac{n+2}{2}e^\lambda + \frac{n}{2}e^{-\lambda}}\right) = \log\sqrt{n(n+2)}.$$ Uniqueness is guaranteed by tracing the tightness of each inequality. The heart of the argument lies in proving the aforementioned monotonicity and is discussed in Section~\ref{sec:technical-KL}.

With all these tools at hand, we can finally solve the $\KL$ projection problem as follows.
\begin{theorem}
For even $n$, the optimal value of \eqref{eq:klmin} is $\log(n+1)-\log\sqrt{n(n+2)}$ with unique solutions as mentioned in \eqref{eq:l2solution}.
\end{theorem}
\begin{proof}
Let $\bm r=\bm p*\bm q$ with $\bm p=(a,1-a)\in \Delta_1,\bm q\in \Delta_{n-1}$. It is clear for $a=0,1$ that $H(\bm r) \le \log n < \log\sqrt{n (n+2)}$. For $0<a<1$, $G_{\bm r}(-a/(1-a))=0$ whence  $$\KL(\bm r\|\bm u) = \log(n+1) - H(\bm r) \stackrel{Lemma~\ref{lem:negative-root-entropy}}{\ge} \log\frac{n+1}{\sqrt{n(n+2)}}.$$
Computing $\KL(\bm r^*\|\bm u)$ directly for $\bm r^*=\bm p^**\bm q^*$ as given by \eqref{eq:l2solution} shows that the above inequality is tight. Therefore, the optimal value of \eqref{eq:klmin} is indeed what is claimed. The uniqueness of solutions follows from the uniqueness claim of Lemma~\ref{lem:negative-root-entropy}.
\end{proof}

\subsection{$\ell_1$ and line-metric $W_1$.}
For this part we focus on the two polyhedral metrics in question. First we start with $\ell_1$ projection, namely,
\begin{equation}
\label{eq:l1problem}\min_{\bm r\in \cM_{2,n}} \norm{\bm r-\bm u}{1}.
\end{equation}
As before we will write $\bm r = \bm p * \bm q$ where $\bm p=(a,1-a)\in \Delta_1, \bm q\in \Delta_{n-1}$. As we shall see, the optimal solution in this case will be non-symmetric in $\bm p$ so there will be at least two solutions in the $(\Delta_1,\Delta_{n-1})$ space namely $(\bm p,\bm q)$ and $(\bm p^{rev},\bm q^{rev})$, where $\bm s^{rev}$ denotes the reversal of the vector. We will prove that these are the only two solutions. 

Focus on $a\ge 1/2$ since reversal gives the other solution and parameterize $\bm p$ with $t=(1-a)/a\in[0,1]$ whence $\bm p = (1/(1+t), 1-1/(1+t))$. We first fix $t$ and optimize over $\bm q\in \Delta_{n-1}$. By Lemma~\ref{lem:fixed-t-l1}, this optimal value and solution, for each fixed $t$, are \begin{align}
\label{eq:fixed-t-l1-distance}2\phi(t) &\sett \frac{2(1+t^{n+1})}{(n+1)(1+t)^2}.\\
\label{eq:fixed-t-minimizer-l1}\bm q(t) &\sett \frac1{n+1}\left(1+t, 1-t^2, 1+t^3, \cdots, 1+t^{n-1}, \frac{2+t-t^n}{1+t} \right)
\end{align}

It remains to optimize over $t\in[0,1]$. Differentiating gives $$\phi'(t) = \frac{(n+1)t^n + (n-1)t^{n+1} - 2}{(n+1)(1+t)^3}.$$
The numerator of this expression of $\phi'(t)$ is clearly increasing over $t\in[0,1]$, is negative at $t=0$ and positive at $t=1$ for $n\ge 2$. Therefore there is a unique solution $\tau_n\in (0,1)$ for the equation $(n+1)t^n + (n-1)t^{n+1} = 2$. This means $\phi$ decreases on $[0,\tau_n]$ and increases on $[\tau_n,1]$ whence $\tau_n$ is the unique minimizer of $\phi$ over $[0,1]$. This gives the following theorem.

\begin{theorem}\label{thm:l1-minimizer}
Let $n$ be even and $\phi, \tau_n$ be as above. Let $a^* = 1/(1+\tau_n), \bm p^*=(a^*,1-a^*)$ and $\bm q^* = \bm q(\tau_n)$ from \eqref{eq:fixed-t-minimizer-l1}. Then the optimal value of \eqref{eq:l1problem} is $2\phi(\tau_n)$ and the two solutions in the $\Delta_1\times\Delta_{n-1}$ space are precisely $(\bm p^*, \bm q^*)$ and their reversals.
\end{theorem}

Now let's move on to the line-metric Wasserstein$-1$ distance on $\set{0,1,\cdots,n}$. From Section~\ref{sec:w-dist-tutorial}, we have \begin{equation}\label{eq:w1-expressions}W_1^{\text{line}}(\bm r,\bm u) = \inf_{\gamma\in \Gamma(\bm r,\bm u)}\E{\gamma} \abs{i-j} = \sup_{\substack{\bm x\in \R^{n+1}\\\abs{x_i-x_{i+1}}\le 1}} \sum_{i=0}^{n} \left(r_i-\frac{1}{n+1}\right)x_i = \sum_{k=0}^{n-1}\abs{\sum_{i=0}^kr_i-\frac{k+1}{n+1}}.\end{equation}
To state again, we want to solve \begin{equation}
\label{eq:w1-minimize}
\min_{\bm r\in \cM_{2,n}} W_1^{\text{line}}(\bm r,\bm u).
\end{equation}
The technique for solving this problem will be the same as that of $\ell_1$: we will find a feasible point in each of the primal and dual problems with matching objectives. It turns out surprisingly that this has the same optimal solutions as the $\ell_1$ case. We include the proof of this theorem here to help the reader gain a better insight into the proof of the $\ell_1$ case which is similar.
\begin{theorem}\label{thm:w1-minimizer}
For even $n$, the optimal value of \eqref{eq:w1-minimize} is $\phi(\tau_n)$ and its two solutions are $(\bm p^*,\bm q^*)$ and $((\bm p^*)^{rev}, (\bm q^*)^{rev})$ in the $\Delta_1\times \Delta_{n-1}$ space.
\end{theorem}

\begin{proof}
First we verify that the objective for $(\bm r^*)^{rev} = (\bm p^*)^{rev}*(\bm q^*)^{rev}$ is indeed $\phi(\tau_n)$. This is enough because $W_1(\bm \mu, \bm \nu)$ only depends on $\bm \mu-\bm \nu$ and hence does not change if the arguments are reversed. Recall that $\tau_n$ is the unique solution of $(n+1)t^{n} + (n-1)t^{n+1}=2$ in $(0,1)$. Just for the verification, we say $\tau, \tilde{\bm r},\bm \mu,\bm \nu$ instead of $\tau_n,(\bm r^*)^{rev}, (\bm p^*)^{rev}, (\bm q^*)^{rev}$ respectively. Note that \begin{align*}
\bm \mu &= \frac{1}{1+\tau}\left(\tau,1\right)\\
\bm \nu &= \frac{1}{n+1}\left(\frac{2+\tau-\tau^n}{1+\tau}, 1+\tau^{n-1}, 1-\tau^{n-2}\cdots, 1-\tau^2, 1+\tau\right).
\end{align*}
Convolving them gives \begin{align*}
\tilde r_0 &= \frac{\tau_n(2+\tau-\tau^n)}{(n+1)(1+\tau)^2} = \frac{1}{n+1} - \phi(\tau),\\
\tilde r_1 &= \frac{2+2\tau+\tau^2+\tau^{n+1}}{(n+1)(1+\tau)^2} = \frac{1}{n+1} + \phi(\tau),\\
\tilde r_j &= \frac{1}{n+1}\qquad 2\le j \le n.
\end{align*}
Then by the last expression in \eqref{eq:w1-expressions}, $$W_1(\tilde {\bm r},\bm u) = \abs{\frac{1}{n+1}-\phi(\tau)-\frac{1}{n+1}} = \phi(\tau).$$

Now we use the supremum form of of the Wasserstein distance mentioned in \eqref{eq:w1-expressions} to conclude that this is indeed optimal. Let $\bm r$ be any element of $\cM_{2,n}$. So $\bm r = \bm p*\bm q$. Consider $t\sett p_0/p_1 \in [0,1]$ (can be assumed by reversing all distributions) and $x_i^* \sett -(-t)^i/(1+t)$. Then $\abs{x_i-x_{i+1}} = t^i \le 1$. Recall by Theorem~\ref{thm:M2n-membership} that this $t$ satisfies $\sum_{i=0}^n r_i(-t)^i=0$. This particular $\bm x^*$ gives the objective $$\sum_{i=0}^{n}(r_i-(n+1)^{-1})x_i^* = \frac{1}{(1+t)(1+n)}(1-t+t^2-\cdots+t^{n}) = \phi(t).$$ This shows that $W_1^{\text{line}}(\bm r,\bm u) \ge \phi(t)$ for any $t$. 

At a global minimizer, $W_1(r,u)\ge\phi(t)\ge\phi(\tau_n)$ forces $t=\tau_n\in(0,1)$. For $D_k\sett \sum_{i=0}^k(r_i-u_i)$. Since $|x_k^*-x_{k+1}^*|=t^k<1$ for $1\le k\le n-1$, equality forces $D_k=0$ for those indices. Equality at $k=0$ gives $D_0=-\phi(t)$. Thus $r=(u_0-\phi(t),u_1+\phi(t),u_2,\cdots,u_n)$, and the convolution recurrence uniquely determines $\bm q$.

Indeed this argument shows, like the $\ell_1$ case, that for every fixed $t\in[0,1]$, the minimum of $W_1^{\text{line}}((t/(1+t),1/(1+t))*\bm q)$ over $\bm q\in \Delta_{n-1}$ is $\phi(t)$. Therefore, this is tight and the conclusion of the theorem follows.
\end{proof}

\begin{remark}
Note that in the above proof, the root $-t$ of the generating polynomial of $\bm r$ plays an important role. We would be interested in a deeper reason for its involvement in deciding the structure of the solutions forced by this root $-t$.
\end{remark}

\section{Projections in the $n-$by$-n$ model}
$\cM_{2,n}$ is special in the sense that it has the special linear-factor property due to the Bernoulli counterpart. In this section, we focus on the diagonal model $\cM_{n,n}$. This model has no such comparable deconvolution mechanism. The single strong observation that works here is the fact that the middle coordinate of the convolved law is bounded by the geometric mean of its endpoints. This observation is enough to figure out the smooth projections and their equality cases.

\begin{lemma}\label{lem:endpoint-obstruction}
If $\bm r = \bm p*\bm q$ with $\bm p, \bm q \in \Delta_{n-1}$ then $r_{n-1} \ge 2\sqrt{r_0r_{2n-2}}$. Equality holds iff $p_0q_{n-1} = p_{n-1}q_0$ and $p_iq_{n-1-i} = 0~\forall~1\le i\le n-2$.
\end{lemma}
\begin{proof}
All terms are nonnegative. Therefore $$r_{n-1} = \sum_{i=0}^{n-1}p_iq_{n-1-i} \ge p_0q_{n-1} + p_{n-1}q_0 \ge 2\sqrt{p_0q_{n-1}p_{n-1}q_0} = 2\sqrt{r_0r_{2n-2}}.$$
\end{proof}

\begin{remark}
A inequality of this form is also useful for the $\cM_{2,n}$ case and is actually a proof technique for showing that $\cM_{2,n}$ does not contain the uniform law for even $n$. We omitted the discussion on that in the earlier section to focus on the algebraic structure because the algebraic technique gives a better insight into the rigidity of the structure of $\cM_{2,n}$. This inequality based approach for the diagonal model case can be found in the book of \textcite[pp. 3, pp. 64]{Bollobas_2006}.
\end{remark}

Let's now move onto the projection problems. 

\subsection{$\ell_2$ and forward $\KL$.}
These projection problems have already been solved in literature. We provide the projection version of those statements along with appropriate references.

The $\ell_2$ projection case was solved by \textcite[Theorem 1]{asgarliEtAl2024}. Namely, they solve
\begin{equation}
\label{eq:diagonal-l2}\min_{\bm r\in \cM_{n,n}} \norm{\bm r-\bm u}{2}.
\end{equation}
The unique optimal solution in $\cM_{n,n}$ is \begin{equation}
\label{eq:diagonal-l2-optimal}\bm r^* = \frac{1}{2(3n-2)}(2,\underbrace{3,\cdots,3}_{n-2}, 4, \underbrace{3,\cdots,3}_{n-2}, 2).
\end{equation}

\begin{proposition}[{\cite[Theorem 1]{asgarliEtAl2024}}]
For every $n\ge 2$, the optimal value of \eqref{eq:diagonal-l2} is $\left(2(3n-2)(2n-1)\right)^{-1/2}$ with unique optimal solution given in \eqref{eq:diagonal-l2-optimal}. The factorization in the $\Delta_{n-1}\times \Delta_{n-1}$ space is unique upto swapping factors and is given by $$\bm p^* = \frac12(1,0,0,\cdots,0,1), \qquad \bm q^* = \frac{1}{3n-2} (2,\underbrace{3,\cdots,3}_{n-2}, 2).$$
\end{proposition}

Note that in the above for $n=2$, there is only one solution in the $\Delta_{n-1}\times \Delta_{n-1}$ space because the two factors coincide. For $n\ge 3$ there are exactly two solutions related by swapping $\bm p,\bm q$.

Now let's move on to the $\KL$ projection problem 
\begin{equation}
\label{eq:diagonal-kl}\min_{\bm r\in \cM_{n,n}} \KL(\bm r\|\bm u).
\end{equation}

We will invoke the special version of \textcite[Theorem 2.1]{kovacevic2021} namely the following.
\begin{proposition}[{\cite[Theorem 2.1, special case for two summands.]{kovacevic2021}}]
Let $X,Y \in\set{0,1,\cdots,s-1}$ be independent random variables. Then $$\max_{P_X,P_Y} H_2(X+Y) = 1 + w_0/2 + (1-w_0)\log_2(s-2) + h_2(w_0)$$
where $w_0 = 2/\left(\sqrt2(s-2)+2\right)$ and $h_2(x)=-x\log_2x-(1-x)\log_2(1-x)$.
\end{proposition}

Their proof also constructs the attaining pair for the distributions of $X,Y$, upto swapping, namely, \begin{align*}
\mathbb P[X=0] = &\mathbb P[X=s-1] = \frac12,\\
\mathbb P[Y=0] = \mathbb P[Y=s-1] = \frac{w_0}{2}, \quad &\mathbb P[Y=j] = \frac{1-w_0}{s-2} ~ (1\le j \le s-2).
\end{align*}
This proof, combined with Lemma~\ref{lem:endpoint-obstruction}, also supplies the fact that the only possible distributions possible for attainment are the one displayed above and the one where $X,Y$ are swapped. Consider the quantity \begin{equation}
\label{eq:D_n}D_n \sett \sqrt2(n-2)+2.
\end{equation} from the above theorem with $s=n$. Further a direct computation in the above theorem shows that the quantity claimed as maximum simplifies to $1+\log_2(n-2+\sqrt 2)$. The resulting maximum entropy in natural base logarithm gives the value $\log(2(n-2+\sqrt 2)) = \log(D_n\sqrt 2)$. This result directly gives us the projection we are looking for.

\begin{theorem}
The optimal value of \eqref{eq:diagonal-kl} is $\log(2n-1) - \log \left(D_n\sqrt 2\right)$. The optimal solution is unique upto factor swapping in the $\Delta_{n-1}\times \Delta_{n-1}$ and is given by $$\bm p^* = \frac{1}{2}(1,0,0,\cdots0,,1), \qquad \bm q^* = \frac{1}{D_n}\left(1,\sqrt 2,\sqrt 2,\cdots,\sqrt 2,1\right).$$
\end{theorem}
\begin{proof}
It is clear that $\KL(\bm r\|\bm u) = \log(2n-1) - H(\bm r)$. The maximum entropy of $H(\bm r)$ over $\cM_{n,n}$ is $\log(D_n\sqrt 2)$, which is given by the above. Therefore the minimum relative entropy is $\log(2n-1)-\log(D_n\sqrt 2)$. The uniqueness assertion follows from the proof of \cite[Theorem 2.1]{kovacevic2021} and the discussion above.
\end{proof}

\subsection{Line-metric $W_1$.}

Let's start with the $W_1^{\text{line}}$ problem, before moving on to a conjecture about the $\ell_1$ projection. We are interested in \begin{equation}
\label{eq:w1-diagonal-problem}
\min_{\bm r\in \cM_{n,n}} W_1^{\text{line}}(\bm r,\bm u).
\end{equation} 
Recall the expressions for the line-metric Wasserstein$-1$ distance from \eqref{eq:w1-expressions}, with adjusted indices. Here we will use the same proof strategy. Propose a feasible point with the optimal value and show, using the supremum form of the linear program, that this is indeed a lower bound.

We will need an elementary deviation inequality, Lemma~\ref{lem:independent-absolute-deviation}, which states that for independent $X,Y\in[0,m]$, we have $\E{}\abs{X+Y-m} \le m/2 + (\E{}X+\E{}Y-m)^2/(2m)$ with equality iff $\E{}X=\E{}Y$ and at least one of $X,Y$ is supported on $\set{0,m}$. Using this, we compute a lower bound on the $W_1$ distance using the supremum expression in \eqref{eq:w1-expressions}, with adjusted indices, computed against three dual-feasible `potentials'. This would result in $$W_1^{\text{line}}(\bm r,\bm u) \ge \max\set{t,\frac{n-1}{4n-2} -\frac{t^2}{2n-2}}$$ where $t = \abs{ \E{}[X+Y]-n+1}$, as proven in Lemma~\ref{lem:w1-diagonal-lower}. The point where the two expressions in the above maximum meet should suggest the mean of an extremal distribution, which are in search of. Indeed, by solving a quadratic, they meet at \begin{equation}
t_* = (n-1)\left(\sqrt{\frac{2n}{2n-1}}-1\right).\label{eq:w1-diagonal-optimal}
\end{equation} The equality case of Lemma~\ref{lem:independent-absolute-deviation} states that one of the optimal $X,Y$ has to be supported on $\set{0,n-1}$, say $X$, and $u\sett \E{}X=\E{}Y$. This means $\mathbb P[X=n-1](n-1) = u$ and $u=(n-1\pm t_*)/2$ whence $\mathbb P[X=n-1] = (n-1\pm t_*)/(2(n-1))$. Let's go forward with the choice $\mathbb P[X=n-1] = (n-1 - t_*)/(2(n-1))$ whence $\mathbb P[X=0] = (n-1 + t_*)/(2(n-1))$. The other choice simply reverses the distribution of $X$. Here is the full description.

\begin{theorem}\label{thm:w1-optimizer-diagonal}
The optimal value of \eqref{eq:w1-diagonal-problem} is $t_*$ as in \eqref{eq:w1-diagonal-optimal} above. Upto swapping factors and simultaneously reversing each factor, the unique solution in the $\Delta_{n-1}\times \Delta_{n-1}$ space is given by \begin{equation}
\label{eq:w1-optimal-solution-diagonal}\bm p^* = (a,0,0,\cdots,0,1-a), \qquad \bm q^* = (\eta,\eta,\cdots,\eta,b)
\end{equation}
where $a=(n-1+t_*)/(2n-2), \eta = ((2n-1)a)^{-1}$ and $b=1-(n-1)\eta$.
\end{theorem}
\begin{proof}
Lemma~\ref{lem:w1-diagonal-lower} and the above discussion gives $W_1^{\text{line}}(\bm r,\bm u) \ge t_*$ for any $\bm r\in \cM_{n,n}$. Lemma~\ref{lem:w1-diagonal-attain} verifies that the proposed $\bm p^*,\bm q^*$ satisfy this equality. For uniqueness upto swapping and reversal, one simply traces the tight cases of the inequalities used and the equality criterion of Lemma~\ref{lem:independent-absolute-deviation}, by reversing the factors if necessary.
\end{proof}

\subsection{$\ell_1$ projection and its conjectural lower bound}
We now turn our interest in the $\ell_1$ projection problem \begin{equation}
\label{eq:l1-diagonal-problem}
\min_{\bm r\in \cM_{n,n}} \norm{\bm r-\bm u}{1}
\end{equation}

We exhibit a family of probability vectors and their $\ell_1$ distance from $\bm u$ serves as a candidate for the above projection problem. This, in particular, serves as an upper bound for the optimal value of \eqref{eq:l1-diagonal-problem} but a lower bound still needs to be proven. We conjecture that this is the lower bound too.

Consider the family for probability vectors for $1\le t\le 2$
\begin{equation}
\label{eq:l1-candidate}
\bm p = \frac12(1,0,\cdots,0,1)\in \Delta_{n-1}, \qquad \bm q^{(t)} = \frac{1}{2n-1}(t, \underbrace{2,\cdots,2}_{n-2}, 3-t)\in \Delta_{n-1}.
\end{equation}
One can verify easily that the above are indeed probability vectors just by adding the coordinates. It is also easy to check that the $\ell_1$ distance of $\bm p*\bm q^{(t)}$ from $\bm u$ is exactly $(2n-1)^{-1}$. Indeed, the coordinates of $\bm p*\bm q^{(t)}$ are exactly $(2n-1)^{-1}$ everywhere except indices $0,n-1,2(n-1)$ where they are $t/(4n-2), 3/(4n-2), (3-t)/(4n-2)$ respectively. Adding the deviations from $(2n-1)^{-1}$ gives the aforementioned distance. We conjecture that this is optimal.

\begin{conjecture}
The optimal value of \eqref{eq:l1-diagonal-problem} is $(2n-1)^{-1}$ for $n\ge 3$.
\end{conjecture}

\section{Bernoulli additive noise capacity}

Let $B\in\set{0,1}$ be a Bernoulli random variable with $p=\mathbb P[B=1]$. Define $$C(p) \sett \sup_{\substack{X\perp B\\\mathbb P[X \in \mathbb N]=1\\\E{} X<\infty}}\frac{I(X;X+B)}{\E{}X}$$ where the information content $I(X;X+B)$ is given by $I(X;X+B) = H(X+B) + p\log_2 p + (1-p)\log_2(1-p)$. We let $h_2(p) = -p\log_2 p - (1-p)\log_2(1-p)$ and $h(p) = -p\log p - (1-p)\log(1-p)$ (in base $e$). \textcite{BCPR2026} identify the quantity $C(p)$ with the capacity of the binary sticky channel whose input runs are independently extended by $B$; see \cite[Section 1.1 and Appendix A]{BCPR2026}. \cite[Theorem 3.1]{BCPR2026} supplies an exponential-family upper bound, and \cite[Conjecture 3.6]{BCPR2026} asserts its optimality. The next theorem proves the nondegenerate range of that conjecture. The conversion to the base-two normalization is simply multiplying the exponents by a logarithmic factor and the endpoint case of the conjecture is remarked below.

We consider the distribution \begin{equation}\label{dext:eq:capacity-normalization} Q_b(y) \sett \exp{-c(y-p)-h(p)-b(-\rho)^{y-1}}, \qquad y\in \mathbb N\end{equation}
where $\rho = (1-p)/p$. This is not a distribution on $\mathbb N$ in general, but it turns out by Lemma~\ref{lem:solution-law} that there is a solution $c=c_p(b)$ for every $p\in \left(\frac12,1\right), b\in \R$, which makes it sum to one. Further, $c_p$ has a unique positive minimizer $b_p$.

\begin{theorem}
\label{dext:thm:capacity}
Fix $1/2<p<1$, and let $\rho=(1-p)/p$. Then 
\begin{equation}
\label{dext:eq:capacity-value}
 C(p)=c_p(b_p)/\log 2=\min_{b\in\mathbb R}c_p(b)/\log 2.
\end{equation}
The maximizing argument $X$ of $C(p)$ is unique, has strictly positive probability at every positive integer, and has an exponentially decaying tail.
\end{theorem}
\begin{proof}Every base used here is base $e$ (for logarithms, entropy, mutual information). By Lemma~\ref{lem:solution-law} there is a unique solution $c_p$ that normalizes the law $(Q_b(y))_{y\in\mathbb N}$ and $b_p = \arg\min c_p > 0$. Put $c_*=c_p(b_p)$. Stationarity yields
\begin{equation}
\label{dext:eq:channel-root}
 \sum_{y\geq1}Q_{b_p}(y)(-\rho)^{y-1}=0.
\end{equation}
The shifted law $r_j=Q_{b_p}(j+1)$ has the form in Lemma~\ref{lem:positive-deconvolution} (with infinite support), with
$$
 z=e^{-c_*},\qquad A=e^{-c_*(1-p)-h(p)},\qquad b=b_p.
$$
Since $\rho/(1+\rho)=1-p$, the lemma asserts the existence of a unique law $\bm q$ on $\mathbb N \cup \set{0}$ such that $\bm r=(1-p,p)*\bm q$. Taking $X$ to be an independent random variable with law $\mathbb P(X=j+1)=q_j$ produces $Q_{b_p}$. The tail bound on $q_j$ in the lemma ensures $\E{} X<\infty$.

We next justify the entropy calculation for every $X$ in the definition of $C(p)$, not only for the constructed law above. It is clear that a probability law $R$ on the positive integers with finite mean has finite entropy. Thus $H(X)$ and $H(Y)$ are finite. The cross entropy against $Q_{b_p}$ is finite too, because $-\log Q_{b_p}(y)$ is a linear function of $y$ plus a bounded function. Hence all terms in the following relative-entropy identity are well defined.

For every admissible $\tilde X$, letting $\tilde Y=\tilde X+B$, independence and boundedness of the powers of $\rho$ give
$$\E{}(-\rho)^{\tilde Y-1} =\E{}(-\rho)^{\tilde X-1}\left((1-p)-p\rho\right)=0.$$
Since $\E{} \tilde Y-p=\E{} \tilde X$ and $H(\tilde Y\mid \tilde X)=h(p)$, we obtain
\begin{align*}
\KL(\tilde Y \| Y)
&=-H(\tilde Y) - \sum_{i\ge 1}P_{\tilde Y}(i)\log Q_{b_p}(i)\\
&=-H(\tilde Y) +h(p) + c_*\sum P_{\tilde Y}(i)(i-p) + b_p\sum P_{\tilde Y} (i)(-\rho)^{i-1}\\
&= -H(\tilde Y) + h(p) + c_* \E{}[\tilde X] + b_p\sum P_{\tilde Y} (i)(-\rho)^{i-1}\\
&= -H(\tilde Y) + h(p) + c_* \E{}[\tilde X] + b_p \E{}(-\rho)^{\tilde Y-1}\\
&= c_*\E{}[\tilde X] - H(\tilde Y) + H(\tilde Y\mid \tilde X)\\
&= c_*\E{}\tilde X - I(\tilde X;\tilde Y)
\end{align*}
Nonnegativity of $D_{\mathrm{KL}}$ and taking supremum over $\tilde X$ proves $C(p)\leq c_*/\log 2$. For the constructed law $Q_{b_p}$ of $Y$, the relative entropy is zero, so equality is attained and $C(p)=c_*/\log 2$. Finally, equality for any maximizing random variable forces $P_{\tilde X+B}=Q_{b_p}$. The uniqueness conclusion of Lemma~\ref{lem:positive-deconvolution} forces the same distribution. Since the law of $\tilde Y$ is unique, so is the law of $\tilde X$.
\end{proof}

\begin{remark}[Conversion to the stated conjecture base and the endpoint]
\label{dext:rem:capacity-endpoint}
In the notation of \cite[Conjecture~3.6]{BCPR2026}, which uses bits, $b=\beta\log2$ and $c_p(b)=\lambda(p,\beta)\log2$. Thus \eqref{dext:eq:capacity-value} is precisely the asserted formula for $1/2<p<1$. The endpoint $p=1$ must not be included in the assertion that a normalizing $c$ exists for every real $b$, or that the minimum has a finite minimizer. With $0^0=1$, its normalization equation is
$$e^{-b}+\frac{e^{-c}}{1-e^{-c}}=1.$$
It has a finite solution exactly when $b>0$, namely
$$c_1(b)=\log\frac{2-e^{-b}}{1-e^{-b}},\qquad \inf_{b>0}c_1(b)=\log2,$$
and the infimum is approached only as $b\to\infty$. Directly, $B=1$ implies $I(X;X+B)=H(X)$, and comparison with $\mathbb P(X=x)=2^{-x}$ gives $H(X) \le (\log2)\E{} X$, with equality only at the specified law. Hence $C(1)=1$.
\end{remark}

\printbibliography

@article{chenRaoShreve1997,
  author  = {Chen, Guantao and Rao, M. Bhaskara and Shreve, Warren E.},
  title   = {Can One Load a Set of Dice So That the Sum Is Uniformly Distributed?},
  journal = {Mathematics Magazine},
  volume  = {70},
  number  = {3},
  pages   = {204--206},
  year    = {1997},
  doi     = {10.1080/0025570X.1997.11996534}
}

@article{gasarchKruskal1999,
  author  = {Gasarch, William I. and Kruskal, Clyde P.},
  title   = {When Can One Load a Set of Dice so That the Sum Is Uniformly Distributed?},
  journal = {Mathematics Magazine},
  volume  = {72},
  number  = {2},
  pages   = {133--138},
  year    = {1999},
  doi     = {10.1080/0025570X.1999.11996715}
}

@article{morrison2018,
  author  = {Morrison, Ian},
  title   = {Sacks of Dice with Fair Totals},
  journal = {The American Mathematical Monthly},
  volume  = {125},
  number  = {7},
  pages   = {579--592},
  year    = {2018},
  doi     = {10.1080/00029890.2018.1473699}
}

@article{asgarliEtAl2024,
  author  = {Asgarli, Shamil and Hartglass, Michael and Ostrov, Daniel N. and Walden, Byron},
  title   = {A Fair Shake: How Close Can the Sum of {$n$}-Sided Dice Be to a Uniform Distribution?},
  journal = {The American Mathematical Monthly},
  volume  = {131},
  number  = {7},
  pages   = {608--617},
  year    = {2024},
  doi     = {10.1080/00029890.2024.2347163}
}

@article{kovacevic2021,
  author  = {Kova\v{c}evi\'c, Mladen},
  title   = {On the Maximum Entropy of a Sum of Independent Discrete Random Variables},
  journal = {Theory of Probability and Its Applications},
  volume  = {66},
  number  = {3},
  pages   = {482--487},
  year    = {2021},
  doi     = {10.1137/S0040585X97T99054X},
  eprint  = {2008.01138},
  archivePrefix = {arXiv}
}

@article{bauschCubitt2016,
  author  = {Bausch, Johannes and Cubitt, Toby},
  title   = {The Complexity of Divisibility},
  journal = {Linear Algebra and its Applications},
  volume  = {504},
  pages   = {64--107},
  year    = {2016},
  doi     = {10.1016/j.laa.2016.03.041}
}

@article{covenMeyerowitz1999,
  author  = {Coven, Ethan M. and Meyerowitz, Aaron D.},
  title   = {Tiling the Integers with Translates of One Finite Set},
  journal = {Journal of Algebra},
  volume  = {212},
  number  = {1},
  pages   = {161--174},
  year    = {1999},
  doi     = {10.1006/jabr.1998.7628}
}

@article{vallender1974,
  author  = {Vallender, S. S.},
  title   = {Calculation of the Wasserstein Distance Between Probability Distributions on the Line},
  journal = {Theory of Probability and Its Applications},
  volume  = {18},
  number  = {4},
  pages   = {784--786},
  year    = {1974},
  doi     = {10.1137/1118101}
}

@article{kelly1950,
  author  = {Kelly, John B.},
  title   = {Elementary Problem {E925}},
  journal = {The American Mathematical Monthly},
  volume  = {57},
  number  = {6},
  pages   = {416},
  year    = {1950}
}

@misc{klich2026,
  author        = {Klich, Israel},
  title         = {On Factorizing Aggregate Counting Distributions into Independent Latent Processes},
  year          = {2026},
  eprint        = {2607.14409},
  archivePrefix = {arXiv},
  primaryClass  = {math-ph},
  note          = {Version 1, 15 July 2026}
}

@misc{BCPR2026,
      title={Capacity of Additive-Noise Sticky Channels}, 
      author={Cécile Bouette and Samuel Pearson and Roni Con and João Ribeiro},
      year={2026},
      eprint={2608.01433},
      archivePrefix={arXiv},
      primaryClass={cs.IT},
      url={https://arxiv.org/abs/2608.01433}, 
}

@book{villani_ot,
  author    = {Villani, C{\'e}dric},
  title     = {Topics in Optimal Transportation},
  series    = {Graduate Studies in Mathematics},
  volume    = {58},
  publisher = {American Mathematical Society},
  address   = {Providence, RI},
  year      = {2003},
  isbn      = {978-0-8218-3312-4}
}

@misc{computationaloptimaltransport,
      title={Computational Optimal Transport}, 
      author={Gabriel Peyré and Marco Cuturi},
      year={2020},
      eprint={1803.00567},
      archivePrefix={arXiv},
      primaryClass={stat.ML},
      url={https://arxiv.org/abs/1803.00567}, 
}

@book{Bollobas_2006, 
    place={Cambridge}, 
    title={The Art of Mathematics: Coffee Time in Memphis}, 
    publisher={Cambridge University Press}, 
    author={Bollob\'as, B\'ela}, 
    year={2006}}

\newpage
\appendix

\section{Exposition on Wasserstein distance}\label{sec:w-dist-tutorial}

Fix a finite metric space $(S=\set{0,1,\cdots,n-1}, d)$, and two probability densities $\bm \mu, \bm \nu\in \Delta_{n-1}$. A \textit{transportation plan} to transport $\bm \mu$ to $\bm \nu$ is a matrix $\gamma \in \R_{\ge 0}^{n\times n}$ such that $$\sum_{j=0}^{n-1} \gamma_{ij} = \mu_i ~\forall ~0\le i< n,\qquad \sum_{i=0}^{n-1} \gamma_{ij} = \nu_j ~\forall~ 0\le j< n.$$ Intuitively $\gamma_{i,j}$ records the amount of mass transported from site $i$ to site $j$, so that the total mass that can be transported from site $i$ to all $j$, namely the sum of $\Gamma_{ij}$ over all $j$, should be $\mu_i$. Similarly the total mass transported to site $j$ is $\nu_j$. One can also think of $\gamma$ to be a joint distribution on $S\times S$ with marginals $\bm\mu,\bm\nu$. Let $\Gamma(\bm \mu,\bm \nu)$ be the set of all such transportation plans and is said to be the transportation polytope. The $p-$cost of plan $\gamma$ is proportional to $$C_p(\gamma) = \left(\sum_{0\le i,j<n} \gamma_{ij}d_{ij}^p\right)^{1/p} = \left(\E{(i,j)\sim\gamma} \left[d_{ij}^p\right]\right)^{1/p}.$$ The \textit{Wasserstein-$p$ distance to `transport' $\bm\mu$ to $\bm\nu$ based on metric $d$} is the cost for the `best' transportation plan, namely, $$W_p^{(d)}(\bm \mu,\bm \nu) \sett \inf_{\gamma\in \Gamma(\bm\mu,\bm \nu)}C_p(\gamma) = \inf_{\gamma\in \Gamma(\bm\mu,\bm \nu)}\left(\E{(i,j)\sim\gamma} \left[d_{ij}^p\right]\right)^{1/p}.$$

For $p=1$ and a finite metric space, the above $W_1^{(d)}$ turns into a linear program and has a different form by the Kantorovich-Rubenstein duality, which is simply LP-duality for finite dimensional linear programs. This dual program can be written as a supremum over the Lipschitz polytope, namely $$\text{Lip}_{S,d} = \set{\bm f \in \R^S \st \abs{f_i-f_j} \le d_{ij}~\forall~i,j\in S}.$$
The corresponding dual program is
$$W_1^{(d)}(\bm \mu,\bm \nu) = \inf_{\gamma\in \Gamma(\bm\mu,\bm \nu)}\E{(i,j)\sim\gamma} d_{ij} \stackrel{\text{duality}}{=} \sup_{\bm x\in \text{Lip}_{S,d}} \sum_{0\le i<n} (\mu_i-\nu_i)x_i.$$

\textcite[Section 1.1.3]{villani_ot} recounts the following interpretation from Caffarelli:

\begin{quote}
Suppose you want to ship some coal from mines, distributed as $\bm \mu$, to factories, distributed as $\bm \nu$. The cost function of transport is $d$. Now a shipper comes and offers to do the transport for you. You would pay him $f(x)$ per coal for loading the coal at $x$, and pay him $g(y)$ per coal for unloading the coal at $y$.

For you to accept the deal, the price schedule must satisfy  $f(x) +g(y) \le d(x,y)$. The Kantorovich-Rubenstein duality states that the shipper can make a price schedule that makes you pay almost as much as you would ship yourself.
\end{quote}

We focus on the case when $p=1$ and $d_{ij} = \abs{i-j}$, which is said to be the line metric, and write $W_1^{\text{line}}$ for this setting. In this setting, there is in fact a simple formula in terms of the cumulutaive distributions of $\bm\mu,\bm\nu$ as described by \textcite[Remark 2.30, Equation (2.37)]{computationaloptimaltransport}. To the best of our knowledge, this is attributed to \textcite{vallender1974}. Consider the CDF's of $\bm\mu,\bm\nu$ namely $F_{\bm\mu}, F_{\bm\nu}$ where $F_{\bm\pi}(k) = \sum_{i=0}^k\pi_i$. Then $$W_1^{\text{line}}(\bm\mu,\bm\nu) = \sum_{k=0}^{n-2}\abs{F_{\bm\mu}(k) - F_{\bm\nu}(k)}.$$

\section{The negative-root entropy bound}
\label{sec:technical-KL}

We complete the entropy argument used in the $2$-by-$n$ model. Fixing a negative zero $-\rho$ converts the root condition into one linear moment constraint. Gibbs' inequality then identifies the corresponding exponential family, so the problem reduces to the optimized partition function
$$
\mathcal Z(\rho)=\min_{\lambda\in\R}
\sum_{i=0}^{2k}e^{\lambda(-\rho)^i}.
$$ 
Our goal is to show that $\mathcal Z$ increases on $(0,1]$. At its minimizing parameter, the derivative of $\mathcal Z$ is governed by a weighted alternating exponential sum. We determine the sign of that sum by pairing its even and odd terms and comparing the resulting cross-ratios. The only lengthy part is a scalar estimate ensuring that all these cross-ratios point in the same direction.

\begin{proposition}[Gibbs variational inequality]
\label{prop:gibbs}
Let $I$ be finite, let $(f_i)_{i\in I}\subseteq\R$, and let $Z(\lambda)=\sum_{i\in I}e^{\lambda f_i}$. For every probability vector $(r_i)_{i\in I}$ and every $\lambda\in\R$,
\begin{equation}
H(\bm r)\leq\log Z(\lambda)-\lambda\sum_{i\in I}r_if_i.
\label{eq:gibbs}
\end{equation}
Equality holds if and only if $r_i=e^{\lambda f_i}/Z(\lambda)$ for every $i$.
\end{proposition}

\begin{proof}
Set $\pi_i=e^{\lambda f_i}/Z(\lambda)$. Nonnegativity of relative entropy, which follows from the log-sum inequality, gives
$$
0\leq\sum_i r_i\log\frac{r_i}{\pi_i}
=-H(\bm r)-\lambda\sum_i r_if_i+\log Z(\lambda).
$$
Equality holds exactly when $\bm r=\bm\pi$.
\end{proof}

We now prepare the cross-ratio comparison. The parameter
$$
\Lambda=\frac{-2\log\rho}{(1+\rho)(1-\rho^{2k-1})}
$$
is chosen so that the desired cross-ratio sign will follow once the balancing parameter is shown to lie below $\Lambda$. The next estimate establishes that comparison.

\begin{lemma}
\label{lem:scalar-cutoff}
Let $k\geq1$ and $0<\rho<1$, and define $\Lambda$ as in \cref{lem:taylor-comparison}. Then
\begin{equation}
\frac{\rho(1-\rho^{2k})}{1+\rho}
+\Lambda\frac{\rho^2(1-\rho^{4k})}{1-\rho^2}
>e^{-\Lambda}.
\label{eq:scalar-cutoff}
\end{equation}
\end{lemma}

\begin{proof}
First assume $0<\rho\leq1/2$. Note that 
$$
\Lambda(1+r^2) = \frac{-2(1+r^2)\log r}{(1+r)(1-r^{2k-1})} \ge \frac{-2(1+r^2)\log r}{(1+r)} \stackrel{Lemma~\ref{lem:helperbound1}}{>} 1.
$$
Using $\rho^{2k}\leq\rho^2$ and $(1-\rho^{4k})/(1-\rho^2)\geq1+\rho^2$, the left side of \eqref{eq:scalar-cutoff} is at least
$$
\rho-\rho^2+\Lambda\rho^2(1+\rho^2)>\rho.
$$
Moreover $(1+\rho)(1-\rho^{2k-1})<2$, and hence
$$
e^{-\Lambda}
=\rho^{\,2/((1+\rho)(1-\rho^{2k-1}))}<\rho.
$$
This proves the claim in the first range.

Now let's move to the range $1/2\leq\rho<1$. Put $u=\rho^{2k-1}\in[0,\rho]$. \Cref{lem:helperbound2} gives $-(1+\rho)\log \rho > 2(1-\rho)$. Consequently, we have 
\begin{equation}\label{eq:lambda-zero}
\Lambda = \frac{-2\log \rho}{(1+\rho)(1-\rho^{2k-1})} > \frac{4(1-\rho)}{(1+\rho)^2(1-u)} =: \Lambda_0.
\end{equation}
Define
$$
\Psi(L)=
\frac{\rho(1-\rho u)}{1+\rho}
+L\frac{\rho^2(1-\rho^2u^2)}{1-\rho^2}
-\frac{1}{1+L}.
$$
Then
$$\Psi'(L) = \frac{r^2(1-r^2u^2)}{1-r^2} + \frac{1}{(1+L)^2}>0.$$ 
Therefore $\Psi$ is monotonically increasing. We now show $\Psi(\Lambda_0)>0$. Directly substituting $\Lambda_0$ into $\Phi$ gives 
$$
\Psi(\Lambda_0) = \frac{\rho(1-\rho u)}{1+\rho} + \frac{\rho^2(1-\rho^2u^2)}{1-\rho^2} \cdot \frac{4(1-\rho)}{(1+\rho)^2(1-u)} - \frac{1}{1+\frac{4(1-\rho)}{(1+\rho)^2(1-u)}}.$$
This simplifies to $$\Phi(\Lambda_0) = \frac{(\rho-1) P(u,\rho)}{(\rho + 1)^{3} (u - 1) (4 \rho + (\rho + 1)^{2} (u - 1) - 4)}$$ where
\begin{align*}
P(u,\rho)={}&
u^3(3\rho^5+7\rho^4+5\rho^3+\rho^2)\\
&+u^2(-2\rho^5+10\rho^4-3\rho^3+5\rho^2+5\rho+1)\\
&+u(-\rho^5-\rho^4-\rho^3-5\rho^2-6\rho-2)\\
&+1+\rho-17\rho^2-\rho^3.
\end{align*}
Since $u=r^{2k-1}\in[0,r]$, \Cref{lem:dirtypoly} says that $P(u,\rho)<0$. This shows that the numerator of $\Psi(\Lambda_0)$ is positive because both terms in numerator are negative. Similarly, the cubic term in the denominator of $\Psi(\Lambda_0)$ is positive and the other two terms are negative. Indeed $u-1<0$ and $4\rho+(\rho+1)^2(u-1)-4 = 4(\rho-1) + (\rho+1)^2(u-1) < 0$. Therefore $\Psi(\Lambda_0) > 0$.
By \eqref{eq:lambda-zero} and monotonicity of $\Psi$, $\Psi(\Lambda) > \Psi(\Lambda_0) >0$ whence
$$
\frac{\rho(1-\rho^{2k})}{1+\rho} +\Lambda\frac{\rho^2(1-\rho^{4k})}{1-\rho^2} >\frac{1}{1+\Lambda}>e^{-\Lambda}.
$$
\end{proof}

\paragraph{From cross-ratios to an alternating sign}

At the minimizing parameter, the total even and odd weights balance. The next lemma shows that this balance occurs before the threshold $\Lambda$ and hence orders every pair of cross-products in the same way. A symmetrization then turns those pairwise comparisons into the weighted alternating inequality we need.

\begin{lemma}[Cross-term sign]
\label{lem:cross-term}
Let $k\geq1$, $0<\rho<1$, and $\lambda>0$. Define
$$
u_j(\lambda)=\rho^{2j}e^{-\lambda\rho^{2j}}
\quad(0\leq j\leq k),
\qquad
v_j(\lambda)=\rho^{2j+1}e^{\lambda\rho^{2j+1}}
\quad(0\leq j\leq k-1).
$$
If $\sum_{j=0}^ku_j(\lambda)=\sum_{j=0}^{k-1}v_j(\lambda)$, then
\begin{equation}
\lambda< \Lambda:=\frac{-2\log\rho}{(1+\rho)(1-\rho^{2k-1})},
\label{eq:lambda-cross}
\end{equation}
and, for $0\leq j<\ell\leq k$,
\begin{equation}
v_j(\lambda)u_\ell(\lambda) -v_{\ell-1}(\lambda)u_j(\lambda)<0.
\label{eq:cross-term}
\end{equation}
\end{lemma}

\begin{proof}
Let
$$
F(s)=\sum_{j=0}^ku_j(s)-\sum_{j=0}^{k-1}v_j(s).
$$
Every term of $F'(s)$ is negative, so $F$ is strictly decreasing and the hypothesis is $F(\lambda)=0$. Since $u_0(\Lambda)=e^{-\Lambda}$,
\begin{align*}
e^{-\Lambda} - F(\Lambda) &= \sum_{j=0}^{k-1}(v_j(\Lambda) - u_{j+1}(\Lambda))\\
&= \sum_{j=0}^{k-1}\rho^{2j+1}\left[ e^{\Lambda \rho^{2j+1}} - \rho e^{-\Lambda \rho^{2j+2}}\right]\\
&\stackrel{Lemma~\ref{lem:taylor-comparison}}{>} \sum_{j=0}^{k-1} \rho^{2j+1} \left[1-\rho + \Lambda\rho^{2j+1} \left(1+\rho^{2}\right)\right]\\
&= \sum_{j=0}^{k-1} \left( \rho^{2j} (1-\rho)\rho + \Lambda \rho^2(1+\rho^2) \rho^{4j}\right)\\
&= (1-\rho)\rho \cdot\frac{1-\rho^{2k}}{1-\rho^2} + \Lambda \rho^2(1+\rho^2) \cdot \frac{1-\rho^{4k}}{1-\rho^4}\\
&= \frac{\rho(1-\rho^{2k})}{1+\rho} + \Lambda \rho^2 \cdot \frac{1-\rho^{4k}}{1-\rho^2} \stackrel{Lemma~\ref{lem:scalar-cutoff}}{\ge } e^{-\Lambda}
\end{align*}
Thus $F(\Lambda)<0$, and strict monotonicity of $F$ gives \eqref{eq:lambda-cross}.

For $j<\ell$, direct substitution gives
$$
\log\frac{v_{\ell-1}u_j}{v_ju_\ell} =-2\log\rho +\lambda(1+\rho)(\rho^{2\ell-1}-\rho^{2j}).
$$
The difference in the final parentheses is negative, so \eqref{eq:lambda-cross} yields
\begin{align*}
\log\frac{v_{\ell-1}u_j}{v_ju_\ell} &> -2\log\rho\left[ 1-\frac{\rho^{2j}(1-\rho^{2\ell-2j-1})} {1-\rho^{2k-1}}\right]\geq0.
\end{align*}
The final inequality uses $2\ell-2j-1\leq2k-1$ and $\rho^{2j}\leq1$. Exponentiation proves \eqref{eq:cross-term}.
\end{proof}

\begin{theorem}[Weighted alternating sign]
\label{thm:weighted-alternating}
Let $k\geq1$, $0<\rho\leq1$, and $t<0$. If
\begin{equation}
\sum_{i=0}^{2k}(-\rho)^i e^{t(-\rho)^i}=0,
\label{eq:alternating-balance}
\end{equation}
then
\begin{equation}
\sum_{i=0}^{2k}i(-\rho)^i e^{t(-\rho)^i}
\begin{cases}
<0,&0<\rho<1,\\
=0,&\rho=1.
\end{cases}
\label{eq:alternating-weighted}
\end{equation}
\end{theorem}

\begin{proof}
For $\rho=1$, it is straightforward to check that the claimed quantity is zero. So we focus on $\rho\in(0,1)$. We are given $$\sum_{j=0}^k u_j(\lambda) = \sum_{j=0}^{k-1} v_j(\lambda) =:M$$ and want to show that $$\sum_{j=0}^k ju_j(\lambda) < \sum_{j=0}^{k-1} \left(j+\frac12\right) v_j(\lambda).$$
Note that \begin{equation}\label{eq:main}
\sum_{j=0}^{k-1}\left( j+\frac12\right) v_j - \sum_{\ell=0}^{k}\ell u_\ell = \frac1M\sum_{j=0}^{k-1}\sum_{\ell=0}^k \left( j+\frac12-\ell\right) u_\ell v_j.    
\end{equation}
We split the above double sum on the RHS into two parts, one where $j\ge \ell$ and one where $j<\ell$, as follows.
$$\sum_{j=0}^{k-1}\sum_{\ell=0}^k \left( j+\frac12-\ell\right) u_\ell v_j = \sum_{\ell=1}^k \sum_{j=0}^{\ell-1}\left( j+\frac12-\ell\right) u_\ell v_j  + \sum_{\ell=0}^k \sum_{j=\ell}^{k-1} \left( j+\frac12-\ell\right) u_\ell v_j.$$
Notice the second term can be re-wrtten to symmetrize the above sum:
\begin{align*}
\sum_{\ell=0}^k \sum_{j=\ell}^{k-1} \left( j+\frac12-\ell\right) u_\ell v_j &= \sum_{\ell=0}^k \sum_{j=0}^{k-1}  \bm1_{\ell\le j}\left( j+\frac12-\ell\right) u_\ell v_j\\
&= \sum_{j=0}^{k} \sum_{\ell=0}^{k-1}  \bm1_{j\le \ell}\left( \ell+\frac12-j\right) u_j v_\ell\\
&= \sum_{j=0}^{k} \sum_{\ell=1}^{k}  \bm1_{j\le \ell-1}\left( \ell-\frac12-j\right) u_j v_{\ell-1}\\
&= \sum_{\ell=1}^{k}\sum_{j=0}^{\ell-1}  \left( \ell-\frac12-j\right) u_j v_{\ell-1}.
\end{align*}
This gives \begin{equation}\label{eq:paired}\sum_{j=0}^{k-1}\sum_{\ell=0}^k \left( j+\frac12-\ell\right) u_\ell v_j = \sum_{0\le j < \ell \le k} \left( j+\frac12-\ell\right) (u_\ell v_j-u_{j}v_{\ell-1}).\end{equation}
If $j<\ell$, then each term $j+\frac12-\ell < 0$ and $u_\ell v_j- u_jv_{\ell-1} < 0$ by Lemma~\ref{lem:cross-term} proving that the above quantity in \eqref{eq:paired} is positive. But this quantity is a positive scaling (since $M>0$) of $$\sum_{j=0}^{k-1}\left(2j+1\right) v_j - \sum_{j=0}^{k}2ju_j$$ by \eqref{eq:main} and therefore the conclusion of the result follows.
\end{proof}

\paragraph{The optimized partition function}

The alternating sign is exactly the derivative information required for the partition minimum. Strict convexity supplies a unique minimizing parameter, and the envelope identity, or Danskin's theorem, depending on the preference of the reader, differentiates the optimized value without requiring an explicit formula for that parameter.

\begin{lemma}
\label{lem:partition-monotonicity}
For $k\geq1$ and $0<\rho\leq1$, define
$$
Z_\rho(\lambda)=\sum_{i=0}^{2k}e^{\lambda(-\rho)^i},
\qquad
\mathcal Z(\rho)=\min_{\lambda\in\R}Z_\rho(\lambda).
$$
The minimizer $\lambda_\rho$ of $Z_\rho$ is unique and negative. $\mathcal Z$ is strictly increasing on $(0,1]$ with
\begin{equation}
\mathcal Z(1)=2\sqrt{k(k+1)}.
\label{eq:partition-at-one}
\end{equation}
\end{lemma}

\begin{proof}
We have
$$
\partial_\lambda Z_\rho
=\sum_{i=0}^{2k}(-\rho)^ie^{\lambda(-\rho)^i},
\qquad
\partial_{\lambda\lambda}Z_\rho
=\sum_{i=0}^{2k}\rho^{2i}e^{\lambda(-\rho)^i}>0.
$$
Also $\partial_\lambda Z_\rho(0)=(1+\rho^{2k+1})/(1+\rho)>0$, whereas $\partial_\lambda Z_\rho(\lambda)\to-\infty$ as $\lambda\to-\infty$. Thus there is a unique minimizer $\lambda_\rho<0$. The implicit function theorem makes it differentiable in $\rho$. At the minimum, the envelope identity gives
$$
\mathcal Z'(\rho)
=-\lambda_\rho
\sum_{i=0}^{2k}i(-\rho)^{i-1}
e^{\lambda_\rho(-\rho)^i},
$$
and therefore
\begin{equation}
\frac{\rho\mathcal Z'(\rho)}{\lambda_\rho}
=\sum_{i=0}^{2k}i(-\rho)^i
e^{\lambda_\rho(-\rho)^i}.
\label{eq:partition-envelope}
\end{equation}
By \cref{thm:weighted-alternating}, the right side is negative for $\rho<1$ and zero at $\rho=1$. Since $\lambda_\rho<0$, $\mathcal Z'(\rho)>0$ on $(0,1)$. Continuity gives strict increase on $(0,1]$. Finally,
$$
Z_1(\lambda)=(k+1)e^\lambda+ke^{-\lambda}
\geq2\sqrt{k(k+1)},
$$
with equality at $\lambda=\frac12\log(k/(k+1))$, proving \eqref{eq:partition-at-one}.
\end{proof}

\subsection{Completion of the entropy argument} \label{sec:completion-entropy}
Now we finally use the above ingredients to prove the heart of the paper, namely, Lemma~\ref{lem:negative-root-entropy}.

\begin{proof}[Proof of Lemma~\ref{lem:negative-root-entropy}]
Suppose first that $0<\rho\leq1$ and set $f_i=(-\rho)^i$. The root constraint is
$$
\sum_{i=0}^{2k}r_if_i=0.
$$
Applying Proposition~\ref{prop:gibbs} and then minimizing over $\lambda$ gives
$$
H(\bm r)\leq\log\mathcal Z(\rho)
\leq\log\mathcal Z(1)
=\log\bigl(2\sqrt{k(k+1)}\bigr)
=\log\sqrt{n(n+2)}.
$$
If $\rho>1$, reverse the distribution: $\widetilde r_i=r_{2k-i}$. Its generating polynomial is
$$
G_{\widetilde{\bm r}}(x)=x^{2k}G_{\bm r}(1/x),
$$
so it has the root $-1/\rho\in(-1,0)$, while $H(\widetilde{\bm r})=H(\bm r)$. The preceding argument applies.

If the entropy bound is an equality and $\rho\leq1$, strict monotonicity in \cref{lem:partition-monotonicity} forces $\rho=1$. If $\rho>1$, the reversed argument would force $1/\rho=1$, which is impossible. Hence equality requires $\rho=1$. Equality in Gibbs' inequality then forces
$$
r_i=\frac{\exp(\lambda_1(-1)^i)}
{(k+1)e^{\lambda_1}+ke^{-\lambda_1}},
\qquad
\lambda_1=\frac12\log\frac{k}{k+1}.
$$
The denominator is $2\sqrt{k(k+1)}$, and simplification gives $r_i=1/(2k+2)=1/(n+2)$ for even $i$ and $r_i=1/(2k)=1/n$ for odd $i$. Conversely, this vector's generating polynomial has root $-1$, is exactly the displayed Gibbs law, and is tight at every inequality above.
\end{proof}

\section{Supplementary lemmata}

\begin{lemma}\label{lem:quad-reduction}
For every $\bm p\in \Delta_1,\bm q\in \Delta_{n-1}$, $\norm{\bm p*\bm q}{2}^2 \ge \norm{(1/2,1/2)*\bm q}{2}^2$.
\end{lemma}
\begin{proof}
Expansion gives
$$\norm{(a,1-a)*\bm q}{2}^2 = \sum_{i=0}^{n-1} q_i^2 - 2a(1-a)\left(\sum_{i=0}^{n-1}q_i^2 - \sum_{i=1}^{n-1}q_iq_{i-1}\right).$$
The coefficient in the large parenthesis is
$$\sum_{i=0}^{n-1}q_i^2 - \sum_{i=1}^{n-1}q_iq_{i-1} = \frac12\left(q_0^2+q_{n-1}^2 + \sum_{i=1}^{n-1}(q_i-q_{i-1})^2\right) \ge 0.$$
Therefore for any $\bm q\in \Delta_{n-1}$, $\norm{\bm p*\bm q}{2}^2$ decreases as $a(1-a)$ increases. $a(1-a)$ is maximum at $a=1/2$.
\end{proof}

\begin{lemma}\label{lem:fixed-t-l1}
Let $n$ be even and fix $t\in[0,1]$ and let $a=1/(1+t)$. The problem $$\min_{\bm q\in \Delta_{n-1}}\norm{(a,1-a)*\bm q-\bm u}1$$ has the optimal value $2\phi(t)$ in \eqref{eq:fixed-t-l1-distance}, and if $0<t<1$, the unique solution $\bm q(t)$ in \eqref{eq:fixed-t-minimizer-l1}.
\end{lemma}
\begin{proof}
Verifying that $\bm q(t)\in \Delta_{n-1}$ and showing that the $\ell_1$ distance between $\bm u$ and $\bm r(t) = (a,1-a)*\bm q(t)$ is $2\phi(t)$ are straightforward. It is now enough to show that this distance is a lower bound on the objective.

For the lower bound, we frame its LP dual, which is a maximization (of lower bounds) problem and produce a certificate in the dual feasible space which matches the aforementioned objective value. The objective is $\norm{A_t\bm q-\bm u}{1}$ where $A=A_t \in \R^{(n+1)\times n}$ is the matrix corresponding to convolution by $(a,1-a)$. More precisely, $A_{ii}$ are all $a$, $A_{i+1,i}$ are all $1-a$, rest all zero, We leave it to the reader to verify that the following is the LP dual of the said optimization problem over $\bm q\in \Delta_{n-1}$.
\begin{align*}
\max_{\bm y\in \R^{n+1}, \lambda \in \R}\qquad &\lambda-\bm u^\top \bm y\\
\text{s.t. } \qquad & A^\top \bm y\ge \lambda \bm 1\\
& -\bm 1\le \bm y\le \bm 1.
\end{align*}
In general, this is a lower bound on the primal problem. Consider the candidate \begin{align*}\lambda^*&=2a-1, \\
y_n^* = -1, &\qquad y_j^* = \lambda + 2(1-a)\left(\frac{a-1}{a}\right)^{n-1-j} = \lambda + (1-\lambda)(-t)^{n-1-j}~\forall~0\le j<n.\end{align*}
Let's first verify dual feasibility of this candidate. Clearly $\lambda\in [0,1]$ because $1/2\le a\le 1$. Therefore $y_j^*$, for $0\le j\le n-1$. lies between $\lambda^* + (1-\lambda^*) (\pm t^{n-j-1})$. Both these endpoints are a linear combination of $1,\pm t^{n-1-j}$ and $t\le1$. Therefore $\norm{\bm y^*}\infty\le 1$. Note that of $A^\top \bm y^* = \left(ay_0^*+(1-a)y_1^*, \cdots, ay_{n-1}^*+(1-a)y_n^*\right)$. But for $0\le j\le n-2$ we have $$
ay_j^*+(1-a)y_{j+1}^* = \lambda + 2(1-a)\left(\frac{a-1}{a}\right)^{n-j-2}\left[a\cdot\frac{a-1}{a} + 1-a\right] = \lambda^*$$ and $ay_{n-1}^*+(1-a)y_{n}^* = a\lambda^* +2a(1-a) - (1-a) = \lambda^*$. This proves feasibility of the candidate $(\bm y^*,\lambda^*)$. Its objective value is calculated as follows. First $\sum y_i^* = (1+t)^{-2}\left[n(1-t^2) -2t^{n+1} - t^2-1\right]$. So $$\lambda^*-\bm u^\top \bm y^* = \frac{1-t^2}{(1+t)^2} - \frac{n(1-t^2) - 2t^{n+1} - t^2-1}{(n+1)(1+t)^2} = 2\phi(t).$$
This proves the lower bound of the optimization problem.

Now let's move onto uniqueness. $0<t<1 \implies |y_j^*|<1 \quad(0\le j\le n-2)$. Equality in
$$\sum_j|r_j-u_j|\ge\sum_j y_j^*(r_j-u_j)$$
therefore forces $r_j=u_j$ for $0\le j\le n-2$. The convolution recursion determines $q_0,\cdots,q_{n-2}$, and normalization determines $q_{n-1}$. This proves uniqueness.
\end{proof}

\begin{lemma}[A centered absolute-deviation inequality]
\label{lem:independent-absolute-deviation}
Let $m>0$, and let $X,Y$ be independent random variables taking values in $[0,m]$. Then
\begin{equation}
\E{}\abs{X+Y-m}
\leq\frac m2+
\frac{(\E{} X+\E{} Y-m)^2}{2m}.
\label{eq:independent-absolute-deviation}
\end{equation}
Equality holds if and only if $\E{} X=\E{} Y$ and at least one of $X,Y$ is supported on $\set{0,m}$.
\end{lemma}
\begin{proof}
Writing $x=\E{} X$ and $y=\E{} Y$, the identity
$$
X+Y-m=(1-X/m)(Y-m)+(X/m)Y
$$
and the triangle inequality give
\begin{align*}
\E{}\abs{X+Y-m} &\leq\frac1m\E{}\left[m^2-m(X+Y)+2XY\right]\\
&=m-(x+y)+\frac{2xy}{m}\\
&\leq m-(x+y)+\frac{(x+y)^2}{2m},
\end{align*}
where independence is used in the equality and $(x-y)^2\geq0$ in the last line. This is \eqref{eq:independent-absolute-deviation}. Equality in the last line means $x=y$. Pointwise equality in the triangle inequality means that at least one of $X,Y$ is an endpoint. If neither variable were endpoint-supported, independence would give
$$ \mathbb P[0<X<m, 0<Y<m] =\mathbb P[0<X<m]\mathbb P[0<Y<m]>0, $$
contradicting pointwise equality. Conversely, endpoint support of either variable makes the triangle inequality an equality almost surely.
\end{proof}

\begin{lemma} \label{lem:w1-diagonal-lower}
Let $X,Y\in\set{0,1,\cdots,n-1}$ be independent random variables and let their law be $\bm r\in \cM_{n,n}$. Then  $$W_1^{\text{line}}(\bm r,\bm u) \ge \max\set{t,\frac{n-1}{4n-2} -\frac{t^2}{2n-2}}$$ where $t = \abs{ \E{}[X+Y]-n+1}$.
\end{lemma}
\begin{proof}
Let $m=n-1$ throughout this proof. Let $X,Y,\bm r, t$ as mentioned. Let $S=X+Y$ so that $t=\abs{\E{}S-m}$. Recall from \eqref{eq:w1-expressions} that $$W_1^{\text{line}} = \sup_{\substack{\bm x\in \R^{2m+1}\\\abs{x_i-x_{i+1}}\le 1}} \sum_{i=0}^{2m} \left(r_i-\frac{1}{2m+1}\right)x_i.$$

We use three different Lipschitz $\bm x\in \R^{2m+1}$ in the above to get different lower bounds on the Wasserstein distance.

First take $x_i = i$. This satisfies $\abs{x_i-x_{i+1}}=1$. The corresponding objective is $$\sum_{i=0}^{2m}r_ii-\frac{1}{2m+1} \sum_{i=0}^{2m}i = \E{}S-m.$$

Then taking $x_i = -i$ gives the objective $m-\E{}S$.

Finally take $x_i = -\abs{i-m}$. The objective is
$$\frac{1}{2m+1} \sum_{i=0}^{2m}\abs{i-m} - \sum_{i=0}^{2m}r_i\abs{i-m} = \frac{m(m+1)}{2m+1} - \E{}\abs{S-m} \stackrel{\eqref{eq:independent-absolute-deviation}}{\ge} \frac{m}{4m+2} - \frac{t^2}{2m}.$$

Combining these, we get $$W_1^{\text{line}}(\bm r,\bm u) \ge \max\set{m-\E{}S, \E{}S-m,\frac{m}{4m+2} - \frac{t^2}{2m}} = \max\set{t, \frac{m}{4m+2} - \frac{t^2}{2m}}.$$
\end{proof}

\begin{lemma}]\label{lem:w1-diagonal-attain}
Let $\bm p,\bm q\in \Delta_{n-1}$ be as in \eqref{eq:w1-optimal-solution-diagonal}. Then $$W_1^{\text{line}}(\bm p*\bm q, \bm u) = (n-1)\left(\sqrt{\frac{2n}{2n-1}}-1\right).$$
\end{lemma}
\begin{proof}
Firstly, it is an easy check that $\bm p^*,\bm q^*\in \Delta_{n-1}$. Then, \begin{align*}
\bm r^* = \bm p^*\bm q^* = (\underbrace{a\eta, \cdots, a\eta}_{n-1}, (1-a)\eta+ab, \underbrace{(1-a)\eta,\cdots,(1-a)\eta}_{n-2}, (1-a)b)
\end{align*}
Now we use the last expression of \eqref{eq:w1-expressions}, namely,
\begin{equation}\label{eq:sum-wasserstein}W_1^{\text{line}}(\bm r^*,\bm u) = \sum_{k=0}^{2n-3} \abs{\sum_{i=0}^k r_i-\frac{k+1}{2n-1}}.\end{equation}
The first $n-1$ coordinates of $\bm r^*$ are each $a\eta = (2n-1)^{-1}$ whence the above sum has no contribution for $k=0,\cdots,n-2$. The last coordinate of $\bm q^*$  has lower weight than the rest. Indeed, $n\eta = \sqrt{2n/(2n-1)} > 1$ whence $b<\eta$. Now note that $r^*_{j} < (2n-1)^{-1} ~\forall~ n \le j \le 2n-2$. Indeed, $$(1-a)\eta = \frac{n-1-t_*}{n-1+t_*}\cdot \frac{1}{2n-1} < \frac{1}{2n-1}, \qquad (1-a)b < (1-a)\eta < \frac{1}{2n -1}.$$ 
Therefore the sequence of partial sums of $\bm r^*-\bm u$ is zero initially, then is positive at some point and gradually decreases to zero since $r_i < u_i$ for $i\ge n$ and $\bm 1^\top(\bm r^*-\bm u) = 0$. Therefore $r_{n-1} > u_{n-1}$. In particular this means that every absolute value in \eqref{eq:sum-wasserstein} is non-negative. Therefore, $W_1^{\text{line}}(\bm r^*,\bm u) = \E{S\sim\bm r}[2n-2-S] - \E{U\sim \bm u} [2n-2-U] = \E{}U-\E{}S = n-1- 2(n-1)(1-a) = t_*$.
\end{proof}

\begin{lemma}\label{lem:solution-law}
Refer to \eqref{dext:eq:capacity-normalization}. Fix $p\in \left(\frac12,1\right), \rho=(1-p)/p$. $\forall~b\in \R,~ \exists! c=c_p(b)$ satisfying $$\sum_{y=1}^\infty Q_b(y) = 1.$$
The function $c_p$ is smooth, strictly convex, coercive and $b_p\sett \arg\min c_p > 0$.
\end{lemma}

\begin{proof}
Recall $Q_b(y) \sett \exp{-c(y-p)-h(p)-b(-\rho)^{y-1}}$. Write $f_y=(-\rho)^{y-1}$ and $d_y=y-p>0$. For fixed $b$, the expression in \eqref{dext:eq:capacity-normalization}, with $c_p(b)$ replaced by $c>0$ then becomes, upto a factor of $e^{h(p)}$,
$$\sum_{y=1}^\infty \exp\set{-cd_y-bf_y}.$$
This is strictly decreasing in $c$. It tends to infinity as $c\downarrow0$ (since $f_y\to0$) and tends to zero as $c\to\infty$ (since it is bounded above by $e^{|b|}e^{-c(1-p)}/(1-e^{-c})$). This proves existence and uniqueness of $c_p(b)$ and that this value is positive. On compact subsets of $\{c>0,\ b\in\mathbb R\}$, the series and all its $c,b$ derivatives are dominated by a fixed geometric sequence times a polynomial in $y$. The implicit function theorem therefore applies and proves smoothness of $c_p$.

Let $Q_b$ denote the law on $\mathbb N$ from \eqref{dext:eq:capacity-normalization}, that is, 
$$Q_b(y) = e^{-h(p)} \exp\set{-c_p(b)d_y - bf_y}$$
and let $Y$ be a random variable with law $Q_b$. Differentiating gives
\begin{equation}
\label{dext:eq:capacity-derivatives}
 c_p'(b)=-\frac{\E{Y\sim Q_b}f_Y}{\E{Y\sim Q_b}d_Y},\qquad c_p''(b)=
 \frac{\E{Y\sim Q_b}\left[(d_Yc_p'(b)+f_Y)^2\right]}{\E{Y\sim Q_b}d_Y} > 0.
\end{equation}
Non-negativity of the above second derivative is clear. We will argue it is strictly positive. Indeed if $c_p''(b) = 0$ then $d_1c_p'(b)+f_1=d_2c_p'(b)+f_2 = 0$ implying that $c_p'(b) = -f_1/d_1 = -f_2/d_2$ but $d_1,d_2>0$ and $f_1>0>f_2$, contradicting the two expressions for $c_p'(b)$. This implies $c_p$ is strictly convex.

The inequalities $Q_b(1)\leq1$ and $Q_b(2)\leq1$ imply
\[
 c_p(b)\geq\frac{-h(p)-b}{1-p},\qquad
 c_p(b)\geq\frac{\rho b-h(p)}{2-p}.
\]
They prove coercivity and hence existence of a unique minimizer $b_p$ of $c_p$. 

Next we show that $c_p$ is decreasing at $0$, which is enough to prove that its minimizer is positive. At $b=0$ the normalized law is given by $$Q_0(y) = e^{-h(p) - c_p(0)(1-p)} \left(\exp\set{-c_p(0)}\right)^{y-1}$$ which is geometric. Consequently (from its mgf)
\[
 \E{Y\sim Q_0}f_Y=\frac{1-z}{1+\rho z}>0, \qquad z=\exp\set{-c_p(0)}.
\]
Equation~\eqref{dext:eq:capacity-derivatives} gives $c_p'(0)<0$, whence $b_p>0$.
\end{proof}

\begin{lemma}\label{lem:taylor-comparison}
Let $r\in \left(0,1\right),\lambda>0,j\in \mathbb Z_{\ge 0}$. Define $$\Lambda = \frac{-2\log r}{(1+r)(1-r^{2k-1})}.$$ Then 
\begin{equation}\label{eq:taylor-comparison}
e^{\Lambda r^{2j+1}} - re^{-\Lambda r^{2j+2}} \ge 1-r+ \Lambda r^{2j+1} \left(1+r^{2}\right).
\end{equation}
\end{lemma}
\begin{proof}
Taking Taylor expansion gives
$$
e^{\Lambda r^{2j+1}} - re^{-\Lambda r^{2j+2}} = \sum_{i\ge 0} \frac{\Lambda^i r^{(2j+1)i}}{i!}\left(1-r^{i+1}(-1)^i\right).
$$
Clearly $1-r^{i+1}(-1)^i > 0$ if $i\ge 2$.
\end{proof}

\begin{lemma}\label{lem:helperbound1}
If $r\in \left(0,\frac12\right]$ then $2(1+r^2)\log r + r + 1 <0.$
\end{lemma}
\begin{proof}
Let $f(r) = 2(1+r^2)\log r + r + 1$. Then $$f'(r) = 2(r^{-1}+r) + 4r\log r + 1, \qquad f''(r) = -2r^{-2}+6+4\log r.$$
$r\le \frac12$ so $f''(r) \le -8+6-4\log 2 = -2-4\log 2 < 0$ on $\left(0,\frac12\right]$. This means $f'$ is decreasing on $\left(0,\frac12\right]$. Therefore $\inf\limits_{\left(0,\frac12\right]} f' = f'(1/2) = 6-2\log 2 > 4 > 0$. This further means $f$ is increasing on $\left(0,\frac12\right]$ and so if $r\in \left(0,\frac12\right]$ then $$f(r) \le f(1/2) = \frac12\left(3-\log 32\right) < \frac12\left(3-3\log 3\right) < 0$$ where we used the fact that $e<3$ in the last inequality.
\end{proof}

\begin{lemma}\label{lem:helperbound2}
If $r\in \left[\frac12,1\right)$ then $(1+r)\log r + 2(1-r) < 0$.
\end{lemma}
\begin{proof}
Let $f(r) = (1+r)\log r + 2(1-r)$. Then $$f'(r) = r^{-1} + 1 + \log r - 2, \qquad f''(r) = -r^{-2} + r^{-1}.$$
Clearly $f''(r) < 0$ for $r<1$ whence $f'$ is strictly decreasing on $\left[\frac12,1\right)$. Therefore $f'|_{\left[\frac12,1\right)} > f'(1) = 0$. Again this means $f$ is strictly increasing on $\left[\frac12,1\right)$. This lets us conclude that if $r\in \left[\frac12,1\right)$ then $$f(r) < f(1) = 0.$$
\end{proof}

\begin{lemma}\label{lem:dirtypoly}
Fix $r\in \left[\frac12,1\right)$ and consider the polynomial \begin{align*}
P_r(u) &= u^{3} \left( 3 r^{5} + 7 r^{4} + 5 r^{3} + r^{2}\right)  + u^{2} \left( - 2 r^{5} + 10 r^{4} - 3 r^{3} + 5 r^{2} + 5 r + 1\right)  \\
&~~~~~~~+ u \left( - r^{5} - r^{4} - r^{3} - 5 r^{2} - 6 r - 2\right)  + 1 + r - 17r^2- r^{3}.
\end{align*}
Then $P_r(u) < 0$ for all $u\in[0,r]$.
\end{lemma}
\begin{proof}
Note that $\partial_u^2 P_r/2 =  u \left(3 r^{5} + 7 r^{4} + 5 r^{3} +  r^{2}\right)\cdot 3  - 2 r^{5} + 10 r^{4} - 3 r^{3} + 5 r^{2} + 5 r + 1$. Now fix $r\in \left[\frac12,1\right)$ and treat $u$ as a variable over $[0,r]$. Note that the coefficient of $u$ in $\partial_u^2 P/2$ is positive because all terms are positive. Finally $\partial_u^2 P_r(0)/2 = - 2 r^{5} + 10 r^{4} - 3 r^{3} + 5 r^{2} + 5 r + 1 > - 2 r^{2} + 10 r^{4} - 3 r + 5 r^{2} + 5 r + 1 = 10r^4 + 3r^2 + 2r + 1 > 0$. Therefore $P_r$ is convex over $[0,r]$ as a function of $u$, whence its maxima over $u\in[0,r]$ is attained at $u\in \set{0,r}$. Evaluating at $u=0$, $P_r(0) = 1 + r - 17r^2- r^{3}  < 1+1-\frac{17}{4} - \frac{1}{8} = -\frac{19}{8} < 0$ because $\frac12\le r < 1$. For the other endpoint $u=r$ we get $P_r(r) = \left( r - 1\right)  \left( r + 1\right)  \left( 3 r^{6} + 5 r^{5} + 17 r^{4} + 2 r^{3} + 21 r^{2} + r - 1\right) $. The large factor $3 r^{6} + 5 r^{5} + 17 r^{4} + 2 r^{3} + 21 r^{2} + r - 1$ is positive because $21r^2-1 >0$ if $r\ge \frac12$ and other terms are positive. Therefore $P_r(r) < 0$ because $r<1$. This proves that $P_r(u) < 0$ for any $u\in [0,r]$.
\end{proof}

\end{document}